\documentclass[aps,prl,twocolumn]{revtex4-1}
\usepackage{amsmath,amssymb,amsthm,amsfonts,graphicx,float}
\usepackage{thm-restate}
\usepackage{mathtools}
\usepackage{enumitem}
\usepackage{algpseudocode}
\usepackage{algorithm}
\newcommand{\tr}{\text{Tr}}
\newcommand{\Exp}{\mathop{\mathbb{E}}}
\newtheorem{lemma}{Lemma}
\newtheorem{claim}{Claim}

\newcommand{\ket}[1]{|#1\rangle}
\newcommand{\bra}[1]{\langle #1|}

\DeclarePairedDelimiter\abs{\lvert}{\rvert}

\usepackage[colorlinks,citecolor=OliveGreen,urlcolor=blue,linkcolor=blue,bookmarks=true]{hyperref}      
\usepackage{nicefrac}
\usepackage[nameinlink,capitalize]{cleveref}
\usepackage[dvipsnames]{xcolor}
\usepackage{tikz}
\usepackage{dsfont}
\usepackage{comment}
\usepackage{qcircuit}

\def\vcentcolon{\mathrel{\mathop\ordinarycolon}}
\begingroup \catcode`\:=\active
  \lowercase{\endgroup
  \let :\vcentcolon
  }

\DeclareMathOperator{\supp}{supp}
\newtheorem{theorem}{Theorem}
\newtheorem*{theorem*}{Theorem}
\newtheorem{definition}[theorem]{Definition} 
\newtheorem{property}{Property}

\begin{document}
\title{Quantum advantage from measurement-induced entanglement \\ in random shallow circuits}
\author{Adam Bene Watts$^{1,2}$}
\author{David Gosset$^{2,3,4}$}
\author{Yinchen Liu$^{2,3,4}$}
\author{Mehdi Soleimanifar$^{5}$}
\affiliation{$^1$ Department of Pure Mathematics, University of Waterloo}
\affiliation{$^2$ Institute for Quantum Computing, University of Waterloo}
\affiliation{$^3$ Department of Combinatorics and Optimization, University of Waterloo}
\affiliation{$^4$ Perimeter Institute for Theoretical Physics, Waterloo}
\affiliation{$^5$ California Institute of Technology}
\begin{abstract}

We study random constant-depth quantum circuits in a two-dimensional architecture.
While these circuits only produce entanglement between nearby qubits on the lattice, long-range entanglement can be generated by measuring a subset of the qubits of the output state. It is conjectured that this \textit{long-range measurement-induced entanglement} (MIE) proliferates when the circuit depth is at least a constant critical value $d^*$.
For circuits composed of Haar-random two-qubit gates, it is also believed that this coincides with a \textit{quantum advantage phase transition} in the classical hardness of sampling from the output distribution.

Here we provide evidence for a quantum advantage phase transition in the setting of random \textit{Clifford} circuits. 
Our work extends the scope of recent separations between the computational power of constant-depth quantum and classical circuits, demonstrating that this kind of advantage is present in canonical random circuit sampling tasks. In particular, we show that in any architecture of random shallow Clifford circuits, the presence of long-range MIE gives rise to an unconditional quantum advantage. In contrast, any depth-$d$ 2D quantum circuit that satisfies a short-range MIE property can be classically simulated efficiently and with depth $O(d)$. 
Finally, we introduce a two-dimensional, depth-$2$, ``coarse-grained" circuit architecture, composed of random Clifford gates acting on $O(\log(n))$ qubits, for which we prove the existence of long-range MIE and establish an unconditional~quantum~advantage.

\end{abstract}
\maketitle

\section{Introduction}
Identifying computational tasks where quantum computers yield an advantage compared to classical ones is a central goal of quantum information science. To make effective progress, one hopes to understand which families of quantum circuits admit efficient classical simulation algorithms, with a focus on quantum circuit architectures that are experimentally feasible in the near-term.

A key task involving quantum circuits is simulating measurement of their output state in the standard basis. In the worst case, classically sampling from this output distribution is intractable---even for the 2D brickwork architecture, and with circuit depth as small as $d=3$~\cite{TerhalConstantDepth2004}.
However, the classical hardness of this problem for circuit architectures composed of \textit{random local gates} remains less well understood. Such random quantum circuits have diverse applications:  they underpin quantum supremacy experiments \cite{google2019supremacy, harrow2017supremacy, Aaronson2017supremacy}, provide benchmarking schemes for near-term quantum devices \cite{boixo2018characterizing, liu2021benchmarking}, and model typical quantum states arising from the dynamics or ground states of locally-interacting quantum~systems~\cite{fisher2023random}.
Random instances of quantum circuits can also help identify \emph{generic} features that relate to the classical hardness of simulating quantum circuits.

Perhaps surprisingly, it has been conjectured that the complexity of sampling from the output distribution of random geometrically local quantum circuits in a two-dimensional architecture undergoes a \emph{quantum advantage phase transition}: polynomial-time classical simulation is believed to be possible if and only if the circuit depth is 
below a constant critical value $d^*$ \cite{NappShallow2022}.
Recent works have proposed an intriguing physical origin for this computational transition, linking it to a change in the amount of \emph{measurement-induced entanglement} (MIE) in random quantum circuits \cite{NappShallow2022, bao2021finite, Liu2022MeasurementInduced}.

Measurement-induced entanglement, which enables quantum protocols such as teleportation and entanglement swapping, occurs when spatially separated subsystems of a many-body state become entangled due to measurements on other subsystems. When combined with unitary dynamics, MIE gives rise to emergent phases with distinctive entanglement patterns \cite{Skinner2019Monitored, Chan2019Monitored, Li2018Monitored}. 
In 2D constant-depth quantum circuits, where only nearby qubits may be entangled prior to measurements, MIE can lead to the formation of entanglement between more distant qubits. This in turn can affect the cost of classically simulating measurement of all qubits in the output state of the quantum circuit.  Based on numerical simulations and heuristic mappings to statistical models, previous works \cite{bao2021finite, NappShallow2022} have suggested that for circuit depth $d < d^*$, MIE results in \emph{short-range} entanglement, making the circuit amenable to classical simulation techniques. 
However, for depths $d \geq d^*$, MIE induces \emph{long-range} entanglement, which poses an obstacle for the efficient classical simulation of random quantum circuits. \cref{fig:numerics} illustrates the change in behaviour of MIE  at the conjectured critical circuit depth $d^{*}=6$ for two-dimensional random Clifford circuits. 
A similar transition in MIE has been recently observed in random 2D circuits composed of one- and two-qubit gates that form a universal gate set \cite{google2023measurementInduced}.

Establishing this quantum advantage of random quantum circuits---and more broadly, a separation between efficient quantum and classical computation---has remained elusive. 
A series of recent works \cite{bravyi2018quantum, bravyi2020quantum, BeneWatts2019Seperation, Coudron2021Certifiable, LeGall2019AverageCase, Grier2020Interactive, caha2022single, caha2023colossal} have provided unconditional results addressing a more limited notion of quantum advantage. 
In these works, a computational problem is introduced that can be solved by geometrically-local constant-depth quantum circuits, but not by any classical probabilistic circuit whose depth grows sub-logarithmically with the input size. These works also provide distributions over instances such that classical shallow circuits cannot even solve an instance in the average case (see e.g., \cite{bravyi2018quantum,LeGall2019AverageCase}).
Existing results along these lines make use of certain specially tailored families of shallow quantum circuits. This raises the question of whether \textit{generic} shallow quantum circuits over a fixed architecture demonstrate a similar quantum advantage. 

In this paper, we answer this question in the affirmative. We describe a family of random quantum circuits whose measurement distributions cannot be simulated with shallow classical circuits. The key to our findings is establishing a connection between the emergence of long-range MIE at the critical depth $d^*$ and the limitations of low-depth classical simulation methods. 

Let us now describe our results in more detail.

\begin{figure*}[t]
\begin{minipage}{0.45\textwidth}
\centering
\begin{tikzpicture}[scale=1.2, transform shape]
    \def\smallSquareSide{1}
    \def\padding{0.125} 

    \def\shift{-0.5} 
    \def\shiftt{-0.25} 

    \def\myfontsize{\fontsize{12pt}{13pt}\selectfont} 

    \foreach \x in {0,4} {
        \draw[teal, thick] (\x+\shift,\shift) -- (\x+\shift,4+\shift);
        \draw[teal, thick] (\shift,\x+\shift) -- (4+\shift,\x+\shift);
    }

    \draw[black, thick, fill=red, fill opacity=0.2] (1.25+\shift, 1.25+\shift) rectangle (3+\shift, 3+\shift);
 \node[magenta] at (1.3+\shift + 0.15, 3+\shift-0.25) {\myfontsize \textcolor{black}{B}};

\draw[|-|, thick] (4+\shift+0.2, 1.25+\shift) -- (4+\shift+0.2,3+\shift);
\node[] at (4+\shift + 0.4, 2.2+\shift) {$L$};
    
    \draw[black, thick, fill=blue, fill opacity=0.2] (2+\shift, 2+\shift) rectangle (2.25+\shift, 2.25+\shift);
     \node[magenta] at (0.1+\shift + 0.15, 3.9+\shift-0.15) {\myfontsize \textcolor{black}{C}};
    
    \foreach \x in {0,1,2, 3} {
        \foreach \y in {0,1,2,3} {
            \pgfmathsetmacro\startX{\x+\shiftt+ \padding+\shiftt}
            \pgfmathsetmacro\startY{\y+\shiftt + \padding+\shiftt}
            \pgfmathsetmacro\pointSpacing{(\smallSquareSide - 2 * \padding) / 3}
            \foreach \i in {0,...,3} {
                \foreach \j in {0,...,3} {
                    \pgfmathsetmacro\pointX{\startX + \i * \pointSpacing}
                    \pgfmathsetmacro\pointY{\startY + \j * \pointSpacing}
                    \ifnum\x=1
                        \ifnum\y=2
                            \ifnum\i=3
                                \ifnum\j=0
                                    \fill (\pointX, \pointY) circle (0.75pt);

                                \else
                                    \fill (\pointX, \pointY) circle (0.75pt);
                                \fi
                            \else
                                \fill (\pointX, \pointY) circle (0.75pt);
                            \fi
                        \else
                            \fill (\pointX, \pointY) circle (0.75pt);
                        \fi
                    \else
                        \fill (\pointX, \pointY) circle (0.75pt);
                    \fi
                }
            }
        }
    }
\end{tikzpicture}
\caption{A 2D grid of qubits, partitioned into three regions $A,B$ and $C$. The purple region indicates the single qubit $A$. The red region indicates the square shielding region $B$ with side length $L$ which in this example is $L=7$. All other qubits are in $C$.}
\label{fig:abc}
\end{minipage}
\hspace{0.1cm}
\begin{minipage}{0.45\textwidth}
    \centering
    \includegraphics[width=0.93\textwidth]{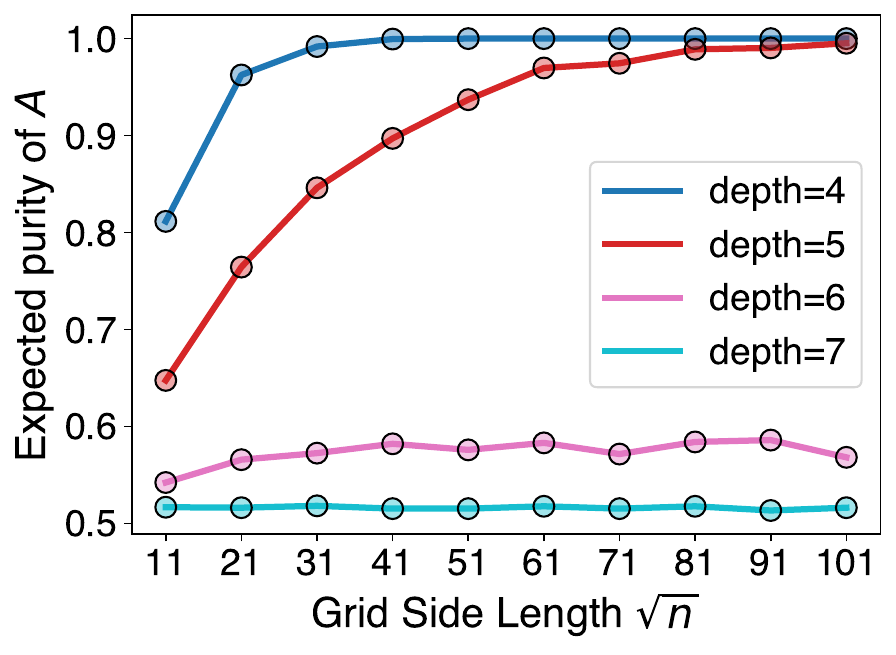}
    \caption{\label{fig:numerics} Classical simulation of MIE in 2D random Clifford circuits in the brickwork architecture (see \cref{fig:brickwork}). Here  $A$ is the qubit at the centre of the grid, $C$ contains all the qubits on the boundary of the grid, and $B$ contains all the other qubits. Each data point is estimated from $1040$ circuit instance samples. The simulation is performed using Stim \cite{gidney2021stim}.}
\end{minipage}
\end{figure*}
\section{Overview of results}
\subsection{Measurement-induced entanglement}
Consider $n$ qubits at the vertices of a 2D square lattice. 
These qubits are partitioned into three sets $A$, $B$, and $C$ where $A$ is a single qubit and $B$ is a square region that shields $A$ from $C$, see \cref{fig:abc}.
All qubits are initialized in the $\ket{0}$ state. We apply a unitary $U$ to $A\cup B\cup C$ and then measure qubits in $B$ in the computational basis. For the moment we consider the case where $U$ is a random depth-$d$ circuit in the brickwork architecture, see \cref{fig:brickwork}.  We study the setting where the side length of $B$ is much larger than $2d$. By a lightcone argument, this ensures that qubit $A$ is unentangled from qubits in $C$ in the output state of the quantum circuit, and only may become entangled as a result of the measurement of qubits in $B$. 

In this setting, one can ask: \emph{What is the shallowest depth $d^*$ at which there is measurement-induced entanglement between $A$ and $C$?}

For simplicity we can imagine that we have already taken the  limit of an infinite lattice $n\rightarrow \infty$. Then for depth $d<d^*$ it is conjectured that  the expected purity $\mathbb{E}[\tr(\rho^2_A)]$ of qubit $A$ (averaged over the random choice of circuit and measurement outcomes on $B$) approaches $1$ exponentially fast as $|B|$ grows, while for $d\geq d^{\star}$ this expected purity is at most $1-\gamma(d)$ where $\gamma(d)>0$ is a constant that may depend on depth $d$, but is independent of the size of the shielding region $B$ \cite{NappShallow2022, bao2021finite}.

This scenario motivates the following definitions of short-range and long-range MIE for output states of 2D random quantum circuits. In the following definitions, a 2D, $n$-qubit, random quantum circuit is described by a probability distribution over quantum circuits that act on $n$ qubits arranged at the vertices of a two-dimensional square grid. A family of 2D random quantum circuits is described by a sequence of such probability distributions, with an increasing number of qubits $n_1,n_2,n_3,\ldots$. Note that here we do not restrict the gates of the circuit to act on nearest-neighbor qubits.

\begin{property}[Short-range MIE]
A family $\mathcal{F}$ of random 2D quantum circuits satisfies the short-range MIE property if the following holds. Suppose $\psi$ is a random pure state on a 2D lattice of $n$ qubits that is generated according to $\mathcal{F}$. We measure all qubits of $\psi$ in an $L\times L$ shielding region $B$ centred at qubit $A\in [n]$, obtaining outcome $x\in \{0,1\}^{|B|}$. Let $\rho_A=\rho_A(x)$ be the postmeasurement state of qubit $A$. Then
\begin{equation}
\mathbb{E}\left[\mathrm{Tr}{\rho_A^2}\right]\geq 1-e^{-\Omega(L)}.
\label{eq:srcondition}
\end{equation}
Here the expectation is taken over the randomness in generating $\psi$ as well as the distribution of random measurement outcomes on qubits in $B$.
\label{prop:shortrangemie}

\end{property}

\begin{property}[Long-range MIE] A family $\mathcal{F}$ of random 2D quantum circuits satisfies the long-range MIE property if there is a constant $c\in (0,1)$ such that the following holds. Let $\psi$ be a random pure state on a 2D lattice of $n$ qubits that is generated according to $\mathcal{F}$.  We measure all qubits of $\psi$ in an $L\times L$ shielding region $B$ centred at qubit $A\in [n]$, obtaining outcome $x\in \{0,1\}^{|B|}$. Let $\rho_A=\rho_A(x)$ be the postmeasurement state of qubit $A$. Here $L$ is any positive integer such that $|C|\neq 0$. Then
\begin{equation}
\mathbb{E}\left[\mathrm{Tr}{\rho_A^2}\right]< c,
\end{equation}
where the expectation is taken over the randomness in generating $\psi$ as well as the distribution of random measurement outcomes on qubits in $B$.
\label{prop:lrmie}
\end{property}

As noted above, we expect 2D brickwork random circuits to have the short-range MIE property for depth $d<d^{*}$ and to satisfy the long-range MIE property for $d\geq d^{*}$. In fact, in the latter case we expect that measurement-induced entanglement is ``everywhere"---that is, present with respect to more general tripartitions $[n]=ABC$ of the qubits (not just the case discussed above where $B$ is a square shielding region centred at $A$). We also expect that this proliferation of MIE gives rise to genuine multipartite entanglement. We shall use the following notion of long-range tripartite MIE, which is phrased in more general terms and can be applied to random circuits in other architectures. 

\begin{property}[Long-range tripartite MIE]
A family $\mathcal{F}$ of random quantum circuits satisfies the long-range tripartite MIE property if there are constants $c_1,c_2\in (0,1)$ such that the following holds. Suppose $\psi$ is a random $n$-qubit pure state that is generated according to $\mathcal{F}$. Let $\Omega=\{h,i,j\}\subseteq [n]$ be a uniformly random triple of qubits. Let $\rho_{\Omega}=\rho_{\Omega}(x)$ be the state of qubits in $\Omega$ after measuring all other qubits of $\psi$ in the computational basis and obtaining outcome $x\in \{0,1\}^{n-3}$. Let $\rho_{h}, \rho_{i}, \rho_{j}$ be its one-qubit reduced density matrices. Then, with probability at least $c_1>0$ we have
\begin{equation}
\mathrm{Tr}[\rho_h^2]<c_2 \quad \text{and} \quad \mathrm{Tr}[\rho_i^2]<c_2 \quad \text{and} \quad \mathrm{Tr}[\rho_j^2]<c_2.
\label{eq:1qentang}
\end{equation}
Here the probability is over the random choice of $\psi$, the choice of triple $\{h,i,j\}$, and the random measurement outcomes on qubits in $[n]\setminus \{h,i,j\}$.
\label{prop:mie}
\end{property}

Rigorously establishing the conjectures described above, and the existence of a phase transition in measurement-induced entanglement for 2D brickwork circuits at a constant depth $d^{*}$, is a challenging open problem. Part of the challenge is that the shallow circuit depth does not allow the quantum system to fully randomize in most senses of interest.  On the other hand, for larger circuit depths $d$ growing with $n$, random circuits can be easier to analyze as they begin to inherit properties of the Haar-random unitary ensemble, such as the unitary $t$-design property and anticoncentration.

Indeed, we show that long-range MIE exists whenever there is  subsystem anticoncentration--a feature of random quantum circuits which is expected to hold for depth $\Omega(\log(n))$ \cite{Dalzell2022AntiConcentration}. Let $P_{AB}(0)=\langle 0^{ABC}|U^{\dagger} |0\rangle\langle 0|_{AB} U|0^{ABC}\rangle$ be the probability of measuring all-zeros on qubits in $AB$. Define a measure of subsystem anticoncentration
\begin{equation}
\chi\equiv 4^{|AB|}\mathbb{E}_{U}[P_{AB}(0)^2]-1.
\label{eq:subsystemanti}
\end{equation}
Here $\chi$ depends on the tripartition $ABC$ as well as the random circuit ensemble.  Note that if $U$ is a Haar-random $n$-qubit unitary then 
\[
\chi_{\mathrm{Haar}}=\frac{2^{|AB|}-1}{2^{|ABC|}+1}
\]
which is exponentially small in $|C|$. We say that the subsystem $AB$ anticoncentrates with respect to a random circuit ensemble if $\chi\rightarrow 0$ as $|C|$ grows.

The following result shows that subsystem anticoncentration implies long-range MIE. 

\begin{theorem}[Informal]\label{res:anticoncentration}
For any tripartition of the qubits $ABC$ where $A$ is a single qubit, we have
\[
\mathbb{E}\left[\mathrm{Tr}{\rho_A^2}\right]\leq \frac{1}{2}+\frac{3}{2}\sqrt{\chi}.
\]
\end{theorem}

We will use this bound to establish long-range MIE in a certain ``coarse-grained" circuit architecture which is introduced below; this describes depth-$2$ circuits with gates acting on $O(\log(n))$ qubits at a time. However, we expect long-range MIE to be present even in constant-depth 2D random quantum circuits composed of two-local gates and with depth above the critical value. Such circuits do not satisfy the subsystem anticoncentration condition considered above. Indeed, for $|C|<n$, the $\chi$ value for depth-$d$ random quantum circuits composed of two-local gates over an arbitrary architecture admits the following lower bound
\begin{align*}
\chi&=\sum_{s\in\{0,1\}^{|AB|}}\mathbb{E}_U[\bra{0^n}U^\dagger (Z(s)_{AB}\otimes I_C)\ket{0^n}^2]-1\\
&\geq\mathbb{E}_U[\bra{0^n}U^\dagger Z_A U\ket{0^n}^2]\geq\frac{1}{3}\left(\frac{2}{5}\right)^d
\end{align*}
where the last inequality follows from Lemma 2.12 in Ref.~\cite{liu2021moments}; see also \cite{barak2020spoofing}. Thus, $\chi\geq\Omega(1)$ for $d=O(1)$, demonstrating that a different proof technique is needed to establish long-range MIE in this setting.

Next we describe how MIE is related to classical simulation and quantum advantage.

\subsection{Shallow and efficient classical simulation}

We obtain a simple, efficient, and shallow classical simulation of 2D quantum circuits 
under the condition that they only generate short-range MIE. We expect that this condition is satisfied by 2D brickwork random circuits in  the low-depth regime $d < d^*$ where MIE is expected to be short-range. 

Our classical simulation algorithm is a modified version of the recently introduced ``gate-by-gate" method \cite{bravyi2022simulate}.
\begin{theorem}[Informal]
    Consider a depth-$d$ quantum circuit $U=U_tU_{t-1}\ldots U_1$ acting on $n$ qubits arranged at the vertices of a 2D grid. Suppose that each gate in the circuit is either a CNOT gate between nearest-neighbor qubits on the grid, or a single-qubit gate. Suppose that the output state of any subcircuit of $U$ satisfies the short-range MIE condition \cref{eq:srcondition}. Then there is an efficient algorithm which samples from the output distribution of $U$. Moreover, this probabilistic classical algorithm can be parallelized to depth $O(d)$ using gates that act on $O(\log^2(n))$ bits at a time. 
\label{res:sim}
\end{theorem}

The classical circuit in the above can be parallelized to depth $O(d)$ even in a ``programmable" setting where it takes as input the specification of the individual gates $U_1,U_2,\ldots, U_t$ in the circuit and outputs a binary string approximately sampled from the output distribution. 

Our algorithm complements an existing technique from Ref.~\cite{NappShallow2022} for classically simulating low-depth 2D circuits which is efficient under a different criterion for short-range MIE. The `patching algorithm' described in  \cite{NappShallow2022} proceeds by first sampling from small disjoint subregions of the 2D lattice and then patching them together to obtain an overall sample from the output distribution. 
This algorithm succeeds when the short-range MIE property holds for subregions $A$ of size $O(\log^2(n))$ and when $X = A \cup B \cup C$ corresponds to a subregion $X \subset [n]$ obtained by tracing out disconnected squares of side $O(\log(n))$ from the lattice.  
In contrast, our simulation algorithm for the quantum circuit $U$ in \cref{res:sim} relies on any subcircuit of $U$ exhibiting short-range MIE when $A$ is a single qubit and  $A \cup B \cup C = [n]$ is a tripartition of the entire lattice. 

\begin{figure}
    \centering
    \begin{tikzpicture}[scale=0.76, transform shape]
    \def\gridSize{6} 
    \def\vertexSize{2pt}
    \def\edgeThickness{1.5pt} 
    \def\vertexGap{0.15} 
    \colorlet{myRed}{Thistle}
    \colorlet{myGreen}{teal}
    \colorlet{myBlue}{black}
    \colorlet{myBrown}{Mahogany}
    \foreach \i in {1,...,\gridSize} {
        \foreach \j in {1,...,\gridSize} {
            \fill[black] (\i, \j) circle (\vertexSize); 
        }
    }
    \foreach \j in {1,...,\gridSize} {
        \foreach \k in {1,...,\gridSize} {
            \ifodd\k
                \draw[myRed, line width=\edgeThickness, rounded corners] (\j, \k + \vertexGap) -- (\j, \k+1 - \vertexGap);
            \fi
        }
    }
    \foreach \j in {1,...,\gridSize} {
        \foreach \k in {1,...,\numexpr\gridSize-1} {
            \ifodd\k\else
                \draw[myGreen, line width=\edgeThickness, rounded corners] (\j, \k + \vertexGap) -- (\j, \k+1 - \vertexGap);
            \fi
        }
    }
    \foreach \j in {1,...,\numexpr\gridSize-1} {
        \foreach \k in {1,...,\gridSize} {
            \ifodd\j
                \draw[myBlue, line width=\edgeThickness, rounded corners] (\j + \vertexGap, \k) -- (\j+1 - \vertexGap, \k);
            \else
                \draw[myBrown, line width=\edgeThickness, rounded corners] (\j + \vertexGap, \k) -- (\j+1 - \vertexGap, \k);
            \fi
        }
    }
    \draw[black, thick] (0.5, 0.5) rectangle (\gridSize+0.5, \gridSize+0.5);
\end{tikzpicture}
\caption{2D brickwork architecture on a $6\times6$ grid of qubits. Edges with the same color indicate two-qubit gates that can be applied simultaneously in a single layer. A circuit of depth $d$ applies a sequence of $d$ layers that rotates through the four colors.}
\label{fig:brickwork}
\end{figure}
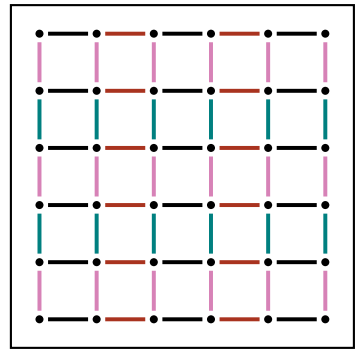

\subsection{Quantum advantage with random shallow Clifford circuits}
Next we specialize to shallow Clifford circuits. These circuits are capable of outperforming classical constant-depth circuits at certain tasks---the kind of quantum advantage described in Refs. \cite{bravyi2018quantum, bravyi2020quantum, BeneWatts2019Seperation, Coudron2021Certifiable, LeGall2019AverageCase, Grier2020Interactive, caha2022single, caha2023colossal}. More precisely, these works use classically controlled (or ``programmable") Clifford circuits: the quantum circuit takes as input a binary string $x$ and then samples from the output distribution of a Clifford circuit $C_x$ that depends on the input. 

To study quantum advantage in random Clifford circuits we view them as programmable Clifford circuits with random inputs. That is, for any depth-$d$ circuit architecture with gates acting on $k$-qubits at a time, we consider a controlled Clifford circuit that takes as input a specification of the individual gates in the circuit. Note that a $k$-qubit Clifford gate can be specified by $O(k^2)$ bits. So the programmable Clifford circuit is also depth-$d$ but has gates that act on $O(k^2)$ qubits at a time.

Can this programmable Clifford circuit, with random input, be simulated by a shallow classical circuit?  For a given choice of gates (input) the classical simulator succeeds if its output distribution is close in total variation distance to the correct output distribution. We say that a classical probabilistic circuit simulates random Clifford circuits in this architecture if it succeeds with high probability over the choice of random gates.

The following result shows that shallow random Clifford circuit families with $k=O(1)$-qubit gates that generate long-range tripartite MIE cannot be simulated by constant-depth classical circuits using gates that act on $K=O(1)$ bits at a time.

\begin{theorem}[Informal]\label{res:quantum-advantage} Consider a depth-$d$ circuit architecture with random Clifford gates. Suppose that the family of random Clifford circuits in this architecture satisfies the long-range tripartite MIE property (\cref{prop:mie}). Suppose $n$ is sufficiently large. Then a probabilistic classical circuit composed of gates with fan-in $K$ and circuit depth $D$ which simulates $n$-qubit random Clifford circuits in this architecture satisfies $D \geq \frac{1}{4}\frac{\log(n)}{\log(K)}$.
\end{theorem}

Since it is conjectured that shallow Clifford circuits in a 2D brickwork architecture with depth $d\geq d^{*}$ satisfy the long-range tripartite MIE property, this result provides evidence that such random shallow Clifford circuits attain a quantum advantage over classical shallow circuits.

The proof of \cref{res:quantum-advantage} extends the classical lower bound from Ref.~\cite{bravyi2018quantum} which can be viewed as exploiting the fact that tripartite MIE in 2D grid graph states is hard for a classical shallow circuit to simulate. While that proof made use of the structure of 2D grid graph states, here we show that long-range tripartite MIE is all that is needed for quantum advantage.

Note that \cref{res:quantum-advantage} applies to more general circuit architectures beyond the 2D brickwork architecture. Can we rigorously establish long-range MIE, long-range tripartite MIE,  and quantum advantage in some architecture of random 2D circuits?

To this end, we introduce a ``coarse-grained'' architecture  shown in \cref{fig:coarse-grained circuit}. The qubits are arranged on a $\sqrt{n} \times \sqrt{n}$ two-dimensional lattice and the circuit applies only two layers of gates.
Each gate acts on a square $\tau\times \tau$ subgrid of qubits. Note that here the depth of the quantum circuit is $2$ and $\tau$ is a parameter that describes the locality of gates. This parameter plays the same role as circuit depth in the 2D brickwork architecture---it measures the linear size of the lightcone of a single qubit \footnote{In particular, if $U$ is a coarse-grained circuit and $O_j$ is an operator acting on qubit $j$, then the support of $UO_jU^{\dagger}$ is contained within a square region of size $O(\tau^2)$ centered at $j$. }.

  \begin{theorem}
Random coarse-grained Clifford circuits with $\tau = O(\sqrt{\log(n)})$ and with $n$ sufficiently large, satisfy both the long-range MIE property, and the long-range tripartite MIE property.
\label{thm:mie_coarsegrained}
\end{theorem}

To prove \cref{thm:mie_coarsegrained} we first show that the coarse-grained architecture with the stated choice of $\tau$ satisfies a strong form of subsystem anticoncentration in which the parameter $\chi$ from \cref{eq:subsystemanti} decays exponentially with the size of $C$. This subsystem anticoncentration is established using a (first moment) probabilistic method. We then show that for Clifford circuits this strong form of subsystem anticoncentration is enough to guarantee both long-range and long-range tripartite MIE.

As a direct corollary of \cref{res:quantum-advantage,thm:mie_coarsegrained} we obtain an unconditional quantum advantage with random Clifford circuits in the coarse-grained architecture:
\begin{theorem}
Random coarse-grained Clifford circuits with $\tau = O(\sqrt{\log(n)})$ and with $n$ sufficiently large, cannot be simulated by any probabilistic classical circuit of depth $D < \frac{1}{4}\frac{\log(n)}{\log(K)}$ and fan-in~$K$.
\label{thm:advantage_coarsegrained}
\end{theorem}

The family of random coarse-grained Clifford circuits with $\tau = O(\sqrt{\log(n)})$ can be implemented (in a programmable sense) via a depth-$2$ quantum circuit with gates acting on $k=O(\log(n)^2)$ qubits. In contrast, \cref{thm:advantage_coarsegrained} states that any classical probabilistic circuit composed of gates with fan-in $K=\mathrm{poly}{\log(n)}$ that simulates the random circuit family must have depth $D>\Omega(\log(n)/\log\log(n))$.

\subsection{Discussion}
We have shown that an unconditional quantum advantage with random shallow Clifford circuits follows from the presence of long-range tripartite MIE. We also exhibited a random circuit architecture with this property.

Moreover, \cref{res:sim,res:quantum-advantage} provide evidence that 2D random shallow Clifford circuits in the standard brickwork architecture undergo a quantum advantage phase transition  that coincides with the emergence of long-range (and long-range tripartite) MIE at the critical depth $d^*=6$. This provides a scaled-down Clifford counterpart to the conjectured phase transition in the classical hardness of simulating random circuits composed of two-qubit Haar-random gates in the same architecture. 

There are several open questions related to this work. Can we prove the existence of a phase transition in MIE at a constant critical circuit depth $d^*$ for 2D brickwork random shallow circuits? Can we show that the quantum advantage with shallow circuits follows from the long-range MIE property (rather than the tripartite version)? Can we establish a quantum advantage in circuit depth for random circuits composed of Haar-random gates (instead of Clifford gates)? 

In the remaining sections of the paper we prove Theorems 1-4, respectively.

\section{Subsystem anticoncentration implies long-range MIE}
\label{sec:finite_clifford_entanglement}
Here we show that random circuits that anticoncentrate also exhibit long-range measurement-induced entanglement. In particular, we prove \cref{res:anticoncentration}.  

Suppose $U=U_{m}U_{m-1}\ldots U_1$ is an $n$-qubit quantum circuit where each gate $U_j$ is a Haar random $k$-qubit gate acting on the qubits in its support.
Here our results apply to general quantum circuits---i.e.,  we do not specify an architecture or require the circuit to act on a system of qubits in a 2D geometry. In the following we shall relate properties of $U$ to those of a random Clifford circuit $D=D_mD_{m-1}\ldots D_1$ where each $D_i$ is a random Clifford unitary that acts on the qubits in the support of $U_i$. 

For any $s\in\{0,1\}^n$, let $Z(s)=\prod_{j=1}^n Z_j^{s_j}$. Here $Z_j$ is the Pauli $Z$ operator acting on the $j$th qubit. Let $[n]=ABC$ be a tripartition of the qubits such that~$|A|=1$. For the output state of the Clifford circuit $D$, the presence or absence of post-measurement entanglement between $A$ and $C$ after measuring $B$ is determined by a set $S$ defined as follows (this set indexes ``entanglement-killing" stabilizers of the state $D|0^n\rangle$). For $P=X,Y,Z$, let
\begin{align}
S_P\equiv &\{s\in \{0,1\}^n: DZ(s)D^{\dagger}=(-1)^\alpha P_A\otimes Z(r)_B\otimes  I_C \nonumber\\ &\text{ for some } r\in \{0,1\}^{|B|}, \alpha\in \{0,1\}\}
\label{eq:defS1}
\end{align}
and let 
\begin{equation}
S=S_X\cup S_Y\cup S_Z.
\label{eq:defS2}
\end{equation} 

\begin{lemma}
\label{lem:post-measurement_criterion_finite}
Suppose $|\phi\rangle=D|0^n\rangle$ where $D$ is a Clifford unitary. Let $[n]=ABC$ be a tripartition of the qubits where $|A|=1$. There is post-measurement entanglement between $A$ and $C$ after measuring all qubits of $B$ in the state $\phi$, if and only if $S=\emptyset$. 
\end{lemma}
\begin{proof}
Let $G=\{DZ(s)D^{\dagger}: s\in \{0,1\}^n\}$ be the stabilizer group of $\phi$.

First suppose $S$ is nonempty. Then $(-1)^{\alpha} P_a\otimes Z(r)_B\otimes I_C$ is a stabilizer of $\phi$, for some $P\in \{X,Y,Z\}$, $r\in \{0,1\}^n$, $\alpha\in \{0,1\}$. After measuring all qubits in $B$ in the computational basis, the postmeasurement state is stabilized by $\pm P_A$ for some choice of sign $\pm 1$ determined by the measurement outcomes. In other words, the postmeasurement state is a pure 1-qubit stabilizer state.

Next suppose we measure all qubits in $B$ , obtain measurement outcome $x\in \{0,1\}^{|B|}$, and the postmeasurement state $\rho_A$ on qubit $A$ is pure. The stabilizer group of this state is generated by
\[
G'=\{(-1)^{x_j} Z_j: j\in B\}\cup G''
\]
where
\[
G''=\{P\in G: P \text{ is $Z$-type on $B$} \}.
\]
(We say that an $n$-qubit Pauli operator is \emph{$Z$-type} if it is a tensor product of $I$ and $Z$ Pauli operators.) Since $\rho_A$ is pure, it must have a stabilizer $\pm P_A\otimes I_{BC}$ for some $P\in \{X,Y,Z\}$. This stabilizer can be expressed as a product of elements from the set of generators above. That is, for some $T\subseteq G''$ and $R\subseteq B$ we have
\[
P_A=\prod_{Q\in T}Q \prod_{j\in R}(-1)^{x_j}Z_j.
\]
From the above we conclude that 
\[
(-1)^{\sum_{j\in R} x_j} P_A\otimes Z(R)_B\otimes I_C=\prod_{Q\in T}Q =DZ(s)D^{\dagger}
\]
for some binary string $s\in \{0,1\}^n$. Therefore $s\in S$ and $S\neq \emptyset$.

\end{proof}

\cref{lem:post-measurement_criterion_finite} shows that emptiness of $S$ determines postmeasurement entanglement for \textit{Clifford} circuits. We now show how this criterion can be extended to our random non-Clifford circuit $U$ of interest.
 
Define $\Pi(x)=|x\rangle \langle x|_B\otimes I_{AC}$ for $x\in \{0,1\}^{|B|}$. If we measure all qubits in $B$, we obtain $x$ with probability
\[
q(x)\equiv \langle 0^n|U^{\dagger}\Pi(x) U|0^n\rangle.
\]
The postmeasurement state on $A$ is then
\[
\rho(x)=\frac{1}{q(x)} \mathrm{Tr}_{BC}\left(\Pi(x) U|0^n\rangle \langle 0^n|U^{\dagger} \Pi(x)\right).
\]
\begin{theorem}[Formal version of \cref{res:anticoncentration}]
The expected purity of the postmeasurement state on $A$ is upper bounded as
\begin{equation}
\mathbb{E}_U \mathbb{E}_{x\sim q} \mathrm{Tr}(\rho(x)^2)\leq \frac{1}{2}+\frac{1}{2} \sqrt{3\mathbb{E}_D[|S|]}.
\label{eq:pur}
\end{equation}
\label{thm:anticonc}
\end{theorem}
Note that the LHS describes the purity of the postmeasurement state in the universal circuit $U$, while the RHS only involves an average over the Clifford circuit $D$. Moreover, this average is upper bounded as
\begin{align}
&\frac{1}{3}\mathbb{E}_D[|S|]=\mathbb{E}_D[|S_Z|]\nonumber\\
&\leq \mathbb{E}_D\left[\sum_{s\in \{0,1\}^{AB}, s\neq 0^{|AB|}}\langle 0^n|D^{\dagger} Z(s)_{AB}\otimes I_C D|0^n\rangle^2\right]\nonumber\\
&=4^{|AB|} \mathbb{E}_D\left[|\langle 0^n|D^{\dagger}(|0\rangle\langle 0|_{AB}\otimes I_C) D |0^n\rangle|^2\right]-1\nonumber\\
&=4^{|AB|} \mathbb{E}_U\left[|\langle 0^n|U^{\dagger}(|0\rangle\langle 0|_{AB}\otimes I_C) U |0^n\rangle|^2\right]-1
\label{eq:anti}
\end{align} 
where in the last line we used the two-design property of the Clifford group. Here \cref{eq:anti} is the quantity denoted $\chi$ in the Introduction, which measures anticoncentration of subsystem $AB$.

\begin{proof}[Proof of \cref{thm:anticonc}]
Write
\[
\alpha_P(x)\equiv \frac{1}{2}\langle 0^n|U^{\dagger} \Pi(x)\cdot\left(P_A\otimes I_{BC}\right) U|0^n\rangle,
\]
so that
\[
\rho(x)=\frac{1}{2\alpha_I(x)}\sum_{P\in \{I,X,Y,Z\}} \alpha_P(x) P.
\]
Using the fact that $\sum_{x}\Pi(x)\otimes \Pi(x)=2^{-|B|} \sum_{s\in \{0,1\}^{|B|}} Z(s)\otimes Z(s)$ we get
\begin{align*}
&\sum_{x\in \{0,1\}^{|B|}} |\alpha_P(x)|^2\\
&=\frac{1}{4\cdot 2^{|B|}}\sum_{s\in \{0,1\}^{|B|}} |\langle 0^n|U^{\dagger} (P_A\otimes Z(s)_B) U|0^n\rangle|^2.
\end{align*}
Then
\begin{align*}
&\mathbb{E}_U \mathbb{E}_{x\sim q}\left(\frac{|\alpha_P(x)|}{2\alpha_I(x)}\right)
=2^{|B|}\mathbb{E}_U \mathbb{E}_{x\sim \mathrm{unif}} |\alpha_P(x)|\\
&\leq 2^{|B|}\left(\mathbb{E}_U \mathbb{E}_{x\sim \mathrm{unif}} |\alpha_P(x)|^2\right)^{1/2}\\
&=\left(\frac{1}{4}\sum_{s\in \{0,1\}^{|B|}}\mathbb{E}_U  |\langle 0^n|U^{\dagger} (P_A\otimes Z(s)_B) U|0^n\rangle|^2\right)^{1/2}.
\end{align*}
Now we use the fact that every gate in $U$ is drawn from the Haar measure, and the fact that the Clifford group is a two-design. This allows us to replace the average over $U$ with the average over the random Clifford circuit $D$ in the above. Then for $P\in \{X,Y,Z\}$ we have
\begin{align}
&\mathbb{E}_U \mathbb{E}_{x\sim q}\left(\frac{|\alpha_P(x)|}{2\alpha_I(x)}\right)\nonumber
\\
&\leq\left(\frac{1}{4}\sum_{s\in \{0,1\}^{|B|}}\mathbb{E}_D  |\langle 0^n|D^{\dagger} (P_A\otimes Z(s)_B) D|0^n\rangle|^2\right)^{1/2}\nonumber\\
&=\frac{1}{2}\sqrt{\mathbb{E}_D |S_P|}\nonumber\\
&=\frac{1}{2\sqrt{3}}\sqrt{\mathbb{E}_D |S|}.
\label{eq:sd}
\end{align}
Now the expected purity of the postmeasurement state is upper bounded as
\begin{align*}
\mathbb{E}_U \mathbb{E}_{x\sim q} \mathrm{Tr}(\rho(x)^2)&=\mathbb{E}_U \mathbb{E}_{x\sim q} \sum_{P\in \{I,X,Y,Z\}}\frac{2|\alpha_P(x)|^2}{4|\alpha_I(x)|^2}\\
&\leq \mathbb{E}_U \mathbb{E}_{x\sim q} \sum_{P\in \{I,X,Y,Z\}}\frac{|\alpha_P(x)|}{2|\alpha_I(x)|}\\
&=\frac{1}{2}+3\mathbb{E}_U \mathbb{E}_{x\sim q} \frac{|\alpha_Z(x)|}{2|\alpha_I(x)|}
\end{align*}
where we used the fact that $|\alpha_P(x)|\leq \alpha_I(x)$ for all $x\in \{0,1\}^{|B|}$. Plugging \cref{eq:sd} into the above completes the proof.

\end{proof}

\section{Efficient and shallow classical simulation of 2D circuits with short-range MIE}

In this section we prove \cref{res:sim}. In particular, we consider the output state $|\psi\rangle=U|0^n\rangle$ of a 2D, depth $d$, quantum circuit and describe an efficient classical simulation method that approximately samples from its output distribution, assuming that the state after each gate in the circuit has short-range MIE. We will also show that the algorithm can be parallelized to depth $O(d)$. The simulation algorithm is a variant of the gate-by-gate method proposed in Ref.~\cite{bravyi2022simulate}.

Here we assume the circuit is expressed as $U=U_mU_{m-1}\ldots U_{1}|0^n\rangle$ where each $U_t$ is either a $\mathrm{CNOT}$ gate or a single-qubit gate. 

In addition to the output state $\psi$ and its output distribution $P(z)=|\langle z|\psi\rangle|^2$, it will also be convenient to consider states and output distributions associated with subcircuits. That is, define
\[
|\psi_t\rangle=U_tU_{t-1}\ldots U_1|0^n\rangle\quad\text{and}\quad P_t(z)=|\langle z|\psi_t\rangle|^2,
\]
for $0\leq t\leq m$ and $z\in \{0,1\}^n$. 

Now imagine that we measure all qubits in $BC$ of the state $\psi_t$ and obtain outcomes $\omega\in \{0,1\}^{|BC|}$. The probability of this measurement outcome is $P_{t,BC}(\omega)\equiv \sum_{y_{BC}=\omega_{BC}} P_t(y)$. The postmeasurement state on $A$ is
\begin{equation}
\rho_t(\omega)=\frac{1}{P_{t,BC}(\omega)}\langle  \omega|_{BC}|\psi_t\rangle\langle \psi _t|\omega\rangle_{BC}.
\label{eq:rhoz}
\end{equation}

If instead we only measure the qubits in $B$ of the state $\psi_t$, obtaining outcome $\omega_B$, then the postmeasurement state of $A$ is

\begin{equation}
\tilde{\rho}_t(\omega_B)\equiv \frac{1}{P_{t,B}(\omega_B)}\langle \omega_B |\mathrm{Tr}_C\left(|\psi_t\rangle\langle \psi_t|\right)| \omega_B\rangle
\label{eq:tilderho}
\end{equation}
where $P_{t,B}(\omega_B)=\sum_{y: y_B=\omega_B} P_t(y)$. The following claim shows that the expected fidelity between $\rho_t(\omega)$ and $\tilde{\rho}_t(\omega_B)$ is equal to the average purity of the latter state.
\begin{lemma}
\[
\mathbb{E}_{z\sim P_{t,BC}} \left[ F(\rho_t(z), \tilde{\rho}_t(z_B))\right]=\mathbb{E}_{z_B\sim P_{t,B}} \left[\mathrm{Tr}( \tilde{\rho}_t(z_B)^2)\right]
\]
\label{lem:purity}
\end{lemma}
\begin{proof}
Using the definitions (\cref{eq:rhoz,eq:tilderho}) we have
\begin{align*}
&\mathbb{E}_{z\sim P_{t,BC}} \left[ F(\rho_t(z), \tilde{\rho}_t(z_B))\right]\\
&=\sum_{z\in \{0,1\}^{|BC|}}\langle \psi_t| (| z_B\rangle\langle z_B|\otimes |z_C\rangle\langle z_C| )\tilde{\rho}_t(z_B)|\psi_t\rangle\\
&=\sum_{z_B\in \{0,1\}^{|B|}}\langle \psi_t | \left(|z_B\rangle\langle z_B|\otimes I_C\right) \tilde{\rho}_t(z_B)|\psi_t\rangle\\
&=\sum_{z_B\in \{0,1\}^{|B|}}P_{t,B}(z_B)\mathrm{Tr}( \tilde{\rho}_t(z_B)^2).
\end{align*}
\end{proof}
Define conditional probability distributions
\begin{equation}
q_t(x_A|z_Bz_C)\equiv \langle x_A|\rho_t(z_Bz_C)| x_A\rangle
\label{eq:q}
\end{equation}
and
\begin{equation}
 \tilde{q}_t(x_A|z_B)\equiv \langle x_A|\tilde{\rho}_t(z_B)| x_A\rangle.
\label{eq:qtil}
\end{equation}

In the above discussion we fixed a partition $[n]=ABC$ but in our classical simulation algorithm---\cref{alg:sim}---we allow this partition to depend on $t$ in a simple way. In particular, let $\Gamma \subseteq [m]$ be such that $U_t$ is a single-qubit gate if and only if $t\in \Gamma$ (it is a $\mathrm{CNOT}$ otherwise).

 Then for $t\in \Gamma$ we let $A(t)$ be the qubit on which $U_t$ acts, $B(t)$ be all qubits within a square subgrid centred at $A(t)$ with side length $L$, and $C(t)=[n]\setminus(A(t)\cup B(t))$. We write $\mathrm{Square}_L(A)$ for the set of all qubits that are a horizontal or vertical distance at most $L/2$ from qubit $A$, so that $B(t)=\mathrm{Square}_L(A(t))$.
\begin{theorem}[Formal version of \cref{res:sim}]
Let $\tilde{P}$ be the output distribution of \cref{alg:sim}. Recall that $\Gamma\subseteq [m]$ are the indices of one qubit gates in the circuit, and $P_t(x)=|\langle x|U_tU_{t-1}\ldots U_1|\psi\rangle|^2$. Then
\begin{equation}
\|\tilde{P}-P_m\|_1\leq 2\sum_{t\in \Gamma} \left(1-\Exp_{z\sim P_{t,B(t)}} \left[\mathrm{Tr}(\tilde{\rho}_t(z)^2)\right]\right)^{1/2}.
\label{eq:tvd}
\end{equation}
\label{thm:simalg}
\end{theorem}

The proof of \cref{thm:simalg} is given in the Appendix. It is obtained by combining \cref{lem:purity} with a robustness property of the gate-by-gate algorithm established in Ref.\cite{bravyi2022simulate}, which we adapt to our setting.

The left-hand side of \cref{eq:tvd} is the approximation error of our algorithm: the total variation distance between the distribution sampled by our classical simulator, and the true output distribution of the quantum circuit.
We are successful in classically simulating the quantum circuit if we can make the right-hand-side of \cref{eq:tvd} small--- at most $\epsilon=O(n^{-c})$ for a given constant $c\geq 1$. Let us assume that MIE in the circuit is short-ranged in the sense that, at each time $1\leq t\leq m$ in the circuit, the output state $|\psi_t\rangle$ has the short-range MIE property (\cref{prop:shortrangemie}). In particular, for each $t\in \Gamma$ this gives
\begin{equation}
\Exp_{z\sim P_{t,B(t)}} \left[\mathrm{Tr}(\tilde{\rho}_t(z)^2)\right]=1-e^{-\Omega(L)}.
\end{equation}
Plugging this into \cref{eq:tvd} we see that the RHS is $O(me^{-L})$. For polynomial-sized circuits (i.e., $m=O(\mathrm{poly}(n))$) we can make this approximation error at most $\epsilon=O(n^{-c})$ by choosing $L=O(\log(n))$. 

Let us now show that \cref{alg:sim} has $O(\mathrm{poly}(n))$ runtime with this choice of $L$.  The runtime of \cref{alg:sim} is $O(mT(L))$ where $T(L)$ is the cost of classically computing the conditional probability $\tilde{q}_t(y|x_{B(t)})$ in line 9 of the algorithm. Using (\cref{eq:tilderho,eq:qtil}) we see that this conditional probability is a ratio of the two marginal probabilities 
\[
\langle y_A x_B|\mathrm{Tr}_C(|\psi_t\rangle\langle \psi|_t)| y_A x_B\rangle
\]
and $P_{t,B}(x_B)$. Since we are tracing out subsystem $C$, each of these marginals can be computed by removing all gates of $U_tU_{t-1}\ldots U_1$ that lie outside of the so-called lightcone of $AB$. This gives a depth $d$ circuit $V$ acting on qubits of a subgrid with side length $L+2d$. Each of these marginal probabilities can be computed using a classical tensor network algorithm from Ref.~\cite{Aaronson2017supremacy} (see Thm 4.3 of that work), which uses runtime $T(L)=2^{O(d(L+d))}$. Thus, for constant-depth $d=O(1)$ circuits with $m=\mathrm{poly}(n)$ gates and with the choice $L=O(\log(n))$, the runtime of \cref{alg:sim} is $\mathrm{poly}(n)$.

Finally, we observe that \cref{alg:sim}, when applied to simulate a quantum circuit of depth $d$,  can be parallelized to depth $O(d)$ if we allow classical gates that act on $K=O(L^2)$ bits. To see this, let us describe how to parallelize all updates to the bit string $x$ in step 2. of the algorithm for a single layer (depth-$1$ circuit) of quantum gates. To this end note that all updates of $x$ corresponding to gates whose support is within a square region of side length $L$ only depend on bits of $x$ in a larger square region of side length $\sim 3L$. These updates can therefore be performed by a single classical ``gate" that acts on bits in the larger square region of side length $3L$. By partitioning the grid into square regions with side length $3L$ we can update $1/9$ of the
bits of $x$ using a single layer of such gates in parallel. We can then simulate the entire depth-$1$ circuit using a classical circuit of depth $9$.  

\begin{figure}
\begin{minipage}{1\linewidth}
\begin{algorithm}[H]
	\caption{Simulation of circuits with short-range MIE (via the gate-by-gate method~\cite{bravyi2022simulate}) \label{alg:sim}}
	\hspace*{\algorithmicindent} \hspace{-26pt} \textbf{Input:}  An $n$-qubit 2D quantum circuit $U=U_m\cdots U_2 U_1$ over the gate set $\mathrm{CNOT}+\mathrm{SU}(2)$, and an integer $L$.\\
 \hspace*{\algorithmicindent} \hspace{-28pt} \textbf{Output:} $x\in \{0,1\}^n$ with prob. $\tilde{P}(x)$ satisfying \cref{thm:simalg}.
	\begin{algorithmic}[1]
	\State{$x\gets 0^n$}
			\For{$t=1$ to $m$}

			\If{$U_t=\mathrm{CNOT}_{ij}$}
				\State{$x_j\gets x_j\oplus x_i$}
			\ElsIf{$U_t\in SU(2)$ acts on qubit $k\in [n]$}
			  \State{$A(t)\gets k$}
				\State{$B(t)\gets \mathrm{Square}_L(A(t))$}
				\State{$C(t)\gets [n]\setminus A(t)\cup B(t)$}

			\State{Sample $y\in \{0,1\}$ from the conditional probability distribution $\tilde{q}_t(y|x_{B(t)})$}
					\State{$x\gets y_{A(t)}x_{B(t)}x_{C(t)}$}

			\EndIf
			      \EndFor
		\State{\textbf{return} $x$}
		\end{algorithmic}
\end{algorithm}
\end{minipage}
\end{figure}

\section{Long-range tripartite MIE implies unconditional quantum advantage in any architecture}

In this section we prove \cref{res:quantum-advantage}. We begin by describing what we mean by a circuit architecture and the associated family of random Clifford circuits. Then we consider the long-range tripartite MIE property (\cref{prop:mie}) and specialize it to Clifford circuits, using the fact that a three qubit stabilizer state that is entangled with respect to any bipartition of the qubits (Cf. \cref{eq:1qentang}), is locally Clifford equivalent to the GHZ state. Finally, we show random Clifford circuits which satisfy the long-range tripartite MIE property cannot be simulated by shallow probabilistic classical circuits. 

Our proof strategy for lower bounding classical circuits which is based on identifying GHZ-type measurement-induced entanglement is similar to the proof strategy from Ref.~\cite{bravyi2018quantum}. The difference is that here we aim to make this strategy work with Clifford circuits composed of uniformly random gates, and we find that long-range tripartite MIE is \textit{all that we require} for quantum advantage.

\subsection{Random Clifford circuits}
\begin{definition}[$n$-qubit circuit template]
An $n$-qubit circuit template $\mathcal{T}$ is a tuple of subsets $(m_1,m_2,\ldots, m_R)$ such that $m_i\subseteq [n]$ for $1\leq i\leq R$. A quantum circuit in this template is a product $U_RU_{R-1}\ldots U_2U_1$ of gates such that $U_i$ acts nontrivially only on qubits in $m_i$. The template is said to be $k$-local if $|m_i|=k$ for all $1\leq i\leq R$, and it is said to have depth $d$ if any quantum circuit in this template can be implemented with circuit depth $d$.
\end{definition}
\begin{definition}[Circuit architecture]
A $k$-local, depth-$d$ circuit architecture $\mathcal{A}=\{\mathcal{T}_{n}\}_{n\in \Gamma}$ is a sequence of $k$-local, depth-$d$ circuit templates with an increasing number of qubits $\Gamma=\{n_1,n_2,n_3\ldots\}$.
\end{definition}

Any template $\mathcal{T}=(m_1,m_2,\ldots, m_R)$ defines a random Clifford circuit, obtained by choosing the gates $U_1,U_2,\ldots, U_R$ to be uniformly random Clifford gates acting on their support. Similarly, an architecture $\mathcal{A}$ is associated with a famiy of random Clifford circuits.

We will be particularly interested in random Clifford circuits which produce long-range tripartite entanglement in the sense of \cref{prop:mie}. We can specialize this property to Clifford circuits as follows.

We say that a three-qubit stabilizer state $|\Phi\rangle$ is GHZ-type entangled iff it is locally Clifford equivalent to the GHZ state. That is, there exist single-qubit Clifford unitaries $C_1,C_2,C_3$ such that 
\[
C_1\otimes C_2\otimes C_3\ket{\Phi}= \frac{1}{\sqrt{2}}\left(|000\rangle+|111\rangle\right).
\]
It is a well-known fact that any three qubit stabilizer state which is not a product state with respect to any bipartition of the qubits, is GHZ-type entangled. 

If $\phi$ is an $n$-qubit stabilizer state, we say that three qubits $\{h,i,j\}\subseteq [n]$ exhibit GHZ-type MIE iff they share GHZ-type entanglement after measuring all qubits in $[n]\setminus\{h,i,j\}$.
\begin{property}[Long-range tripartite MIE, Clifford version]
The architecture $\mathcal{A}=\{\mathcal{T}_n\}_{n\in \Gamma}$ is said to satisfy the \textit{long-range tripartite MIE property} if there is an absolute constant $c>0$ such that the following holds. Suppose $n\in \Gamma$ and consider the circuit template $\mathcal{T}_n\in \mathcal{A}$. Suppose $U$ is a random Clifford circuit with template $\mathcal{T}_n$, i.e. it is composed of gates $U_1,U_2,\ldots, U_R$ chosen uniformly at random from the set of $k$-qubit Clifford gates. Let $|\psi\rangle=U|0^n\rangle$.  Choose a uniformly random triple of qubits $\{h,i,j\}\subseteq [n]$. Then, with probability at least $c>0$ (over the random choice of $U$ and the triple $h,i,j$), the qubits $h,i,j$ exhibit GHZ-type MIE in the state $\psi$.
\label{prop:ghz}
\end{property}

It is not hard to see that \cref{prop:ghz} is equivalent to \cref{prop:mie} for the random Clifford circuit family defined by an architecture $\mathcal{A}$. The condition \cref{eq:1qentang} in the definition of long-range tripartite MIE property implies that the three-qubit state $\rho_{\Omega}$ is not a product state with respect to any bipartition of its qubits. Since in the case at hand it is also a stabilizer state, it is GHZ-type entangled. 

\subsection{Requirements for classical simulation of GHZ-type MIE}

Let us begin by considering what kinds of classical circuits are capable of reproducing the measurement statistics of quantum states with GHZ-type MIE. 

Suppose $U$ is an $n$-qubit Clifford circuit and $|\psi\rangle=U|0^n\rangle$.  Recall that we say that three qubits $\{h,i,j\}\subseteq [n]$ exhibit GHZ-type measurement-induced entanglement iff the state $\rho_{h,i,j}$ after measuring all other qubits in the computational basis is locally Clifford equivalent to the GHZ state. We can rephrase this definition in terms of the stabilizer group of $\psi$, as follows: three qubits $\{h,i,j\}$ of $\psi$   exhibit GHZ-type measurement-induced entanglement iff,  for some tensor product $C_h\otimes C_i\otimes C_j$ of single-qubit Clifford unitaries on $h,i,j$, the stabilizer group of $C_h\otimes C_i\otimes C_j|\psi\rangle$ contains elements
\begin{align}
g_1&=X_hX_iX_jZ(s)\nonumber \\
g_2&=Z_hZ_i I_jZ(t) \nonumber\\
g_3&=I_h Z_iZ_jZ(u),
\label{eq:g13}
\end{align}
for some binary strings $s,t,u\in \{0,1\}^{n}$ such that $s_r=t_r=u_r=0$ for $r\in \{h,i,j\}$.

In the following Lemma we consider a classical function $m:\{0,1,2\}^{3}\rightarrow \{0,1\}^{n}$ which takes input $b\in \{0,1,2\}^3$ describing measurement bases $X,Y,Z$ for qubits $h,i,j$ respectively, and produces output $m(b)\in \{0,1\}^n$. We are interested in whether this function can simulate measurement of $\psi$ in the given bases $b$ on qubits $h,i,j$ and the computational basis on all other qubits. Let $Q(0),Q(1),Q(2)$ denote the single-qubit Clifford unitaries that change basis from the $Z$ basis to the $X,Y,Z$ bases respectively (i.e., $Q(0)=(X+Z)/\sqrt{2}, Q(1)=(Y+Z)/\sqrt{2}, Q(2)=I$). Let $W(b)=Q(b_h)\otimes Q(b_i) \otimes Q(b_j)$ be the tensor product of single-qubit unitaries that change basis to the one described by $b$. A necessary condition for such a simulation of $\psi$ is that
\begin{equation}
\langle m(b)|W(b)|\psi\rangle \neq 0, \quad b\in \{0,1,2\}^{3}.
\label{eq:possibilistic}
\end{equation}
We show a condition under which \cref{eq:possibilistic} is not possible.
\begin{lemma}
Suppose $h,i,j$ exhibit GHZ-type measurement-induced entanglement in the state $|\psi\rangle=U|0^n\rangle$. Let $m:\{0,1,2\}^{3}\rightarrow \{0,1\}^{n}$ and write $m=m(b)$ with $b=b_hb_ib_j\in \{0,1,2\}^3$. Suppose that each output bit $m_r$ depends on at most one of the input variables $b_h,b_i,b_j$, and that 
\begin{itemize}
\item $m_h$ is independent of $b_i, b_j$
\item $m_i$ is independent of $b_h, b_j$
\item $m_j$ is independent of $b_h, b_i$.
\end{itemize}
Then 
\begin{equation}
\mathrm{Pr}_{b\sim \mathrm{unif}(\{0,1,2\}^3)}\left[ \langle m(b)|W(b)|\psi\rangle= 0\right]\geq 1/27.
\label{eq:fails}
\end{equation}
\label{lem:limit}
\end{lemma}
The proof of \cref{lem:limit} is given in the Appendix.
It says that classical circuits that are capable of simulating GHZ-type measurement induced entanglement must have certain kind of input-output correlations. In the next section we will see that this constraint is even more restrictive when $\psi$ is an $n$-qubit state and many triples of qubits exhibit GHZ-type MIE.

\subsection{Lower bound for classical probabilistic simulation}

In the following we shall view a random Clifford circuit as a fixed circuit with a random input that specifies the gates. That is, we construct a controlled Clifford circuit which takes the specification of all Clifford gates $U_1,U_2,\ldots, U_R$ as an input and then applies the corresponding Clifford circuit to an $n$-qubit register.

For any  $k$-local, $n$-qubit circuit template $\mathcal{T}_n=(m_1,m_2,\ldots, m_R)$ we define a controlled-Clifford circuit as follows. The circuit has a data register of $n$ qubits which are initialized in the $|0^n\rangle$ state. In addition, there is one input register $I_j$ of size $O(k^2)$ for each gate $j\in R$ in the circuit, that contains enough qubits to store the description of an arbitrary $k$-qubit Clifford gate. The $j$th gate in the circuit applies Clifford unitary $C$ on qubits $m_j$ of the data register, controlled on the state $|C\rangle$ of the input register $I_j$. An architecture $\mathcal{A}$ defines a family of controlled Clifford circuits in this way, which we denote $CC(\mathcal{A})$.

Suppose $\mathcal{A}=\{\mathcal{T}_n\}_{n\in \Gamma}$ is a depth-$d$, $k$-local architecture. As noted above we can construct an associated family $CC(\mathcal{A})$ of depth-$d$ quantum circuits composed of $O(k^2)$-local controlled-Clifford gates.  Under what conditions can the input/output behaviour of this family of quantum circuits, with random inputs,  be simulated by a constant-depth classical circuit?

Consider a family of probabilistic classical circuits which implement functions $\{F_n\}_{n\in \Gamma}$ with the following input/output properties:
\begin{enumerate}
\item For each $n\in \Gamma$, $F_n$ takes as input a list $U_1,U_2,\ldots, U_{R}$ of $k$-local Clifford gates that define a quantum circuit $U=U_{R}\ldots U_2U_1$ of depth $d$ in the template $\mathcal{T}_n$.
\item $F_n$ also takes as input a random string $r\in \{0,1\}^{\ell(n)}$ drawn from some probability distribution $q_n$. Here $\ell(n)$ can be any function of $n$ and $q_n$ can be any probability distribution over $\ell(n)$ bit strings.
\item The function outputs a binary string
\[
z=F_n(U_1,U_2,\ldots, U_{R},r)\in \{0,1\}^n.
\]
For ease of notation, we sometimes write $z=F_n(U,r)$.
\end{enumerate}

We shall be interested in the case where each function $F_n$ is implemented by a classical circuit with gates of fan-in $K$ and circuit depth $D$. In this case each output bit $z_j$ of $F_n(U,r)$ depends on at most $K^D$ input bits. In particular this implies that each output bit $z_j$ depends on at most $K^D$ of the gates in the circuit (each gate is specified by $O(k^2)$ bits).

If we draw $r\sim q_n$ and compute $z=F_n(U,r)$ then $z\in \{0,1\}^n$ samples from a distribution that we denote 
\[
P_U(z) \qquad z\in \{0,1\}^n.
\]
Let us also write $Q_U(z)=|\langle z|U|0^n\rangle|^2$ for the true output distribution.

\begin{lemma}
Suppose $\mathcal{A}$ is a $k$-local, depth-$d$ circuit architecture with the long-range tripartite MIE property (\cref{prop:ghz}). Suppose $n\in \Gamma $ satisfies $n\geq (10d)^4$ and let $F_n$ be a function with input/outputs as described above. Suppose further that each output bit of $z=F_n(U_1,U_2,\ldots, U_R,r)$ depends on at most $n^{1/4}$ of the $k$-qubit Clifford gates $U_1,U_2,\ldots, U_R$. Then for any $r\in \{0,1\}^{\ell(n)}$, we have
\[
\mathrm{Pr}_{U_1,U_2,\ldots, U_R \sim  \mathcal{D}_k} \left[ \langle F_n(U,r)|U|0^n\rangle=0\right]\geq \frac{c}{27}-O(n^{-1/4}),
\]
where $c>0$ is the absolute constant from \cref{prop:ghz} and $\mathcal{D}_k$ is the uniform distribution on $k$-qubit Clifford unitaries.
\label{lem:n14}
\end{lemma}

\begin{proof}
In the following we fix $r\in \{0,1\}^{\ell(n)}$ and write $z=F_n(U)=F_n(U,r)$.

Let $U_1,U_2,\ldots, U_R$ be the $k$-qubit gates that describe $U$. Here $R\leq nd/k$. These are arranged in $d$ layers. For each $t\in [n]$ there is exactly one $k$-qubit gate that is the last one to act on qubit $t$. Call this gate $V_t$. Note that we may have $V_t=V_\ell$ for $t\neq \ell$.

For each gate $w\in [R]$, let $\mathcal{L}^{\rightarrow}(U_w)\subseteq [n]$ consist of $\mathrm{supp}(U_w)$ as well as all output bits of $F_n$ that depend on gate $U_w$. Here $\mathrm{supp}(U_w)$ denotes the support of $U_w$ (the 
 $k$ qubits on which it acts nontrivially). Let us consider a bipartite graph $G$ with a vertex for each gate $U_w$ with $1\leq w\leq R$,  and a vertex for each output bit $v\in [n]$, and an edge $\{U_w,v\}$ whenever $v\in \mathcal{L}^{\rightarrow}(U_w)$. Since each output bit depends on at most $n^{1/4}$ input gates, and each output bit can be in the support of at most $d$ gates in the circuit, the maximum degree of any output bit is at most $d+n^{1/4}\leq 1.1n^{1/4}$ (since $n\geq (10d)^4$), and therefore the total number of edges in $G$ is 
\begin{equation}
\sum_{w=1}^{R} |\mathcal{L}^{\rightarrow}(U_w)|\leq 1.1n^{5/4}.
\label{eq:sumofL}
\end{equation}
Let 
\[
\mathrm{Bad}=\{t\in [n]: |\mathcal{L}^{\rightarrow}(V_t)|\geq n^{1/2}\} \quad \text{and} \quad \mathrm{Good}=[n]\setminus \mathrm{Bad}.
\]
Then from \cref{eq:sumofL} we have
\[
|\mathrm{Bad}|\leq 1.1n^{3/4}
\]
and therefore 
\begin{equation}
|\mathrm{Good}|\geq n(1-1.1n^{-1/4}).
\label{eq:sizeofgood}
\end{equation}
Now let us consider a graph $G'$ with vertex set $\mathrm{Good}$ and an edge $\{s,t\}$ iff
\[
\mathcal{L}^{\rightarrow}(V_s)\cap \mathcal{L}^{\rightarrow}(V_t)\neq \emptyset.
\]
The maximum degree $\Delta$ of $G'$ satisfies 
\[
\Delta\leq \max_{t\in \mathrm{Good}}|\mathcal{L}^{\rightarrow}(V_t)|\cdot (d+n^{1/4})\leq  n^{1/2}\cdot 1.1n^{1/4}=1.1n^{3/4}.
\]
Now suppose we choose a triple of qubits $\{h,i,j\}\subseteq [n]$ uniformly at random. From \cref{eq:sizeofgood} we see that the probability that $\{h,i,j\}\subseteq \mathrm{Good}$ is $1-O(n^{-1/4})$. Conditioned on this event, the probability that $\{h,i,j\}$ form an independent set in $G'$ can also be lower bounded as $1-O(n^{-1/4})$. To see this note that since the maximum degree of $G'$ is at most $O(n^{3/4})$, two randomly chosen vertices in $\mathrm{Good}$ have an edge between them with probability at most $O(n^{-1/4})$. Then apply this to the three pairs $\{h,i\},\{i,j\}$ and $\{h,j\}$ and use the union bound.

Let $\mathcal{E}$ be the event that our three qubits $h,i,j$ are an independent set in $G'$. Then we have shown $\mathrm{Pr}[\mathcal{E}]\geq 1-O(n^{-1/4})$. Moreover, if $\mathcal{E}$ occurs then 

\begin{align}
\mathcal{L}^{\rightarrow}(V_h)\cap \mathcal{L}^{\rightarrow}(V_i)&=\emptyset\label{eq:findabc1}\\
\mathcal{L}^{\rightarrow}(V_i)\cap \mathcal{L}^{\rightarrow}(V_j)&=\emptyset\label{eq:findabc2}\\
\mathcal{L}^{\rightarrow}(V_h)\cap \mathcal{L}^{\rightarrow}(V_j)&=\emptyset\label{eq:findabc3}.\end{align}

Now suppose that $U_1,U_2,\ldots, U_R \sim \mathcal{D}_k$ are uniformly random $k$-qubit Clifford gates. Let $|\psi\rangle=U|0^n\rangle$ be the output state of the circuit $U$. Let $\mathcal{E}'$ denote the event that qubits $h,i,j$ exhibit GHZ-type measurement-induced entanglement. That is, event $\mathcal{E}'$ occurs iff there is GHZ-type entanglement between qubits $h,i,j$ after measuring all other qubits of $\psi$ in the computational basis. By the long-range GHZ property, event $\mathcal{E}'$ occurs with probability at least $c>0$. 

By the union bound, the probability that both events $\mathcal{E}$ and $\mathcal{E}'$ occur is at least
\[
\mathrm{Pr}[\mathcal{E} \text{ and } \mathcal{E}']\geq c-O(n^{-1/4}).
\]

Now let us suppose that events $\mathcal{E}$ and $\mathcal{E'}$ both occur. Starting with our circuit $U_1,U_2,\ldots, U_R$ we define a set of $27$ quantum circuits as follows. For each of the qubits $h,i,j$ there are three choices for a single-qubit Pauli basis $X,Y,Z$. Let us index these choices by a tuple $b=b_hb_ib_j\in \{0,1,2\}^{3}$, where $0,1,2$ correspond to $X,Y,Z$ respectively. Recall that we write $Q(0),Q(1),Q(2)$ for the single-qubit Clifford unitaries that change basis from the $Z$ basis to the $X,Y,Z$ bases respectively (i.e., $Q(0)=(X+Z)/\sqrt{2}, Q(1)=(Y+Z)/\sqrt{2}, Q(2)=I$). For $b_h,b_i,b_j\in \{0,1,2\}^3$, let
\begin{align*}
V_h(b_h)&=Q(b_h) V_h \\
V_i(b_i)&=Q(b_i) V_i \\
V_j(b_j)&=Q(b_j) V_j.
\end{align*}
and let $U(b)$ denote the circuit obtained from $U_1,U_2,\ldots, U_R$ by making the replacements $V_h\leftarrow V_h(b_h), V_i\leftarrow  V_i(b_i), V_j\leftarrow V_j(b_j)$.

Define a function $m:\{0,1,2\}^3\rightarrow \{0,1\}^n$ by $m(b)=F_n(U(b))$. From \cref{eq:findabc1,eq:findabc2,eq:findabc3} we see that each output bit $m_j$ depends on at most one of the bits $b_h,b_i,b_j$, and that 
\begin{itemize}
\item $m_h$ is independent of $b_i, b_j$
\item $m_i$ is independent of $b_h, b_j$
\item $m_j$ is independent of $b_h, b_i$.
\end{itemize}
Therefore, from \cref{lem:limit} we conclude that 
\[
\langle F_n(U(b))|U(b)|0^n\rangle=0 \quad \quad \text{for some } b\in \{0,1,2\}^3.
\]
Note that each of the circuits $U(b)$ with $b\in \{0,1,2\}^3$ occurs in our random ensemble with the same probability as $U$. Therefore 
\begin{align}
\mathrm{Pr}_{U_1,U_2,\ldots, U_R \sim  \mathcal{D}_k} \left[ \langle F_n(U)|U|0^n\rangle=0\right]&\geq \frac{\mathrm{Pr}[\mathcal{E} \text{ and } \mathcal{E}']}{27}\nonumber\\
&=c/27-O(n^{-1/4}).\nonumber
\end{align}
\end{proof}

\begin{theorem}[Formal version of \cref{res:quantum-advantage}]
Suppose $\mathcal{A}=\{\mathcal{T}_n\}_{n\in \Gamma}$ is a $k$-local, depth-$d$ circuit architecture with the long-range tripartite MIE property (\cref{prop:ghz}). There exists a positive constant $\alpha >0$ and positive integer $n_0$ such that the following holds. Suppose $n\in \Gamma$ satisfies $n\geq \max\{n_0,(10d)^4\}$. Suppose there is a depth-$D$ probabilistic classical circuit composed of gates of fan-in $K$ with output distribution $P_U$ on input circuit $U$ satisfying
\[
\mathbb{E}_{U_1,U_2,\ldots, U_R \sim  \mathcal{D}_k} \left[\|P_U-Q_U\|_1\right]< \alpha,
\]
where $\mathcal{D}_k$ is the uniform distribution on $k$-qubit Clifford unitaries, and $Q_U(z)=|\langle z|U|0^n\rangle|^2$ is the true output distribution. Then $D\geq \frac{1}{4}\frac{\log(n)}{\log(K)}$.
\end{theorem}
\begin{proof}
Let $F_n$ be the function which describes the classical circuit. Recall that $F_n$ takes as input the gates $U_1,U_2,\ldots, U_R$ in the circuit as well as a random string $r\in \{0,1\}^{\ell(n)}$, and outputs a binary string $z=F_n(U_1,U_2,\ldots, U_R, r)=F_n(U,r)$. 

Suppose $D<\frac{1}{4}\frac{\log(n)}{\log(K)}$. Then every output bit of the classical circuit depends on at most $K^D< n^{1/4}$ input bits. In particular, each output bit depends on at most $n^{1/4}$ of the input gates $U_1,U_2,\ldots, U_R$. Then
\begin{align*}
&\mathbb{E}_{U_1,U_2,\ldots, U_R \sim  \mathcal{D}_k} \bigg(\sum_{x\in \{0,1\}^n} |P_U(x)-Q_U(x)|\bigg)\\
&\geq \mathbb{E}_{U_1,U_2,\ldots, U_R \sim  \mathcal{D}_k} \bigg(\sum_{x:Q_U(x)=0} P_U(x)\bigg)\\
&=\mathbb{E}_{U_1,U_2,\ldots, U_R \sim  \mathcal{D}_k} \mathrm{Pr}_{r} \left[\langle F_n(U,r)|U|0^n\rangle=0\right]\\
&=\mathbb{E}_{r} \mathrm{Pr}_{U_1,U_2,\ldots, U_R \sim  \mathcal{D}_k} \left[\langle F_n(U,r)|U|0^n\rangle=0\right]\\
&\geq \mathbb{E}_{r}\frac{c-O(n^{-1/4})}{27}\\
&= \frac{c-O(n^{-1/4})}{27}
\end{align*}
where $c>0$ is the constant appearing in the definition of \cref{prop:ghz}, and where we used \cref{lem:n14}. We choose $n_0$ large enough so that $O(n^{-1/4})$ is at most $c/2$ and then set $\alpha=c/54$.
\end{proof}

\section{Coarse-grained architecture}
\label{sec:Coarse-grained model}

Here we discuss the ``coarse-grained" two-dimensional circuit architecture where we are able to establish rigorous statements concerning measurement-induced entanglement and quantum advantage. In particular, we prove \cref{thm:mie_coarsegrained}. The coarse-grained architecture shares salient features of shallow quantum circuits in the usual brickwork architecture; most importantly, there is a locality parameter $\tau$ that determines the linear size of the lightcone of a given qubit.

Consider a family of random quantum circuits, shown in \cref{fig:coarse-grained circuit}, acting on $n$ qubits arranged on a $\sqrt{n} \times \sqrt{n}$ grid lattice. Each circuit $U$ in this family consist of two layers of gates $V$ and $W$ applied consecutively such that 
$$
U \ket{0^n} = \underbrace{\left(W_{m^2}\otimes \cdots \otimes W_1\right)}_{W} \cdot \underbrace{\left(V_{(m-1)^2}\otimes \cdots \otimes V_1\right)}_{V} \ket{0^n}.
$$
We refer to the gates $V_1,\dots, V_{(m-1)^2}$ in $V$ as the first layer and $W_1,\dots, W_{m^2}$ in $W$ as the second layer gates.

We shall consider two families of quantum circuits in this architecture. The first is obtained by choosing each gate to be a Haar-random unitary acting on a square of $\tau \times \tau$ qubits such that $m^2 \tau^2 = n$. We will call a circuit $U$ drawn from this family a \emph{random coarse-grained universal circuit}. The second family we consider is obtained by choosing each gate to be a uniformly random $\tau^2$-qubit Clifford unitary. We refer a circuit $D$ drawn from this family as a \emph{random coarse-grained Clifford circuit.}

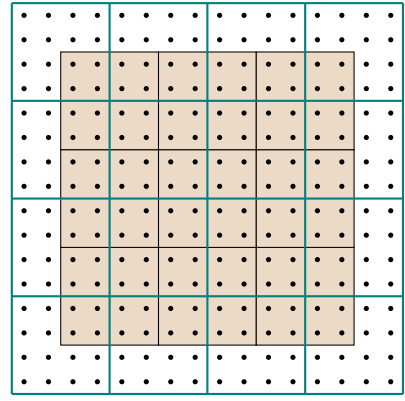
\begin{figure}
\centering
\begin{tikzpicture}[scale=1.3, transform shape]
    \def\smallSquareSide{1}
    \def\padding{0.125}
    
    \def\shift{-0.5} 
    \def\shiftt{-0.25} 

    \foreach \x in {0,1,2} {
        \foreach \y in {0,1,2} {
            \fill[brown!30] (\x,\y) rectangle (\x+\smallSquareSide,\y+\smallSquareSide);
            \draw (\x,\y) rectangle (\x+\smallSquareSide,\y+\smallSquareSide);
        }
    }
    
    \def\myfontsize{\fontsize{7pt}{8pt}\selectfont} 

    \foreach \x in {0,...,4} {
        \draw[teal, thick] (\x+\shift,\shift) -- (\x+\shift,4+\shift);
        \draw[teal, thick] (\shift,\x+\shift) -- (4+\shift,\x+\shift);
    }

    \foreach \x in {0,1,2, 3} {
        \foreach \y in {0,1,2,3} {
            \pgfmathsetmacro\startX{\x+\shiftt+ \padding+\shiftt}
            \pgfmathsetmacro\startY{\y+\shiftt + \padding+\shiftt}
            \pgfmathsetmacro\pointSpacing{(\smallSquareSide - 2 * \padding) / 3}
            \foreach \i in {0,...,3} {
                \foreach \j in {0,...,3} {
                    \pgfmathsetmacro\pointX{\startX + \i * \pointSpacing}
                    \pgfmathsetmacro\pointY{\startY + \j * \pointSpacing}
                    \ifnum\x=1
                        \ifnum\y=2
                            \ifnum\i=3
                                \ifnum\j=0
                                    \fill (\pointX, \pointY) circle (0.75pt);

                                \else
                                    \fill (\pointX, \pointY) circle (0.75pt);
                                \fi
                            \else
                                \fill (\pointX, \pointY) circle (0.75pt);
                            \fi
                        \else
                            \fill (\pointX, \pointY) circle (0.75pt);
                        \fi
                    \else
                        \fill (\pointX, \pointY) circle (0.75pt);
                    \fi
                }
            }
        }
    }
\end{tikzpicture}
\caption{Coarse-grained circuit with two layers of gates for $m=4$ and $\tau=4$. 
The \textcolor{brown!100}{shaded region} denotes the support of the first layer gates $V_1,\dots, V_{9}$. 
The gates with \textcolor{teal}{green boundaries} represent the second-layer gates $W_1,\dots, W_{16}$.}
\label{fig:coarse-grained circuit}
\end{figure}

\subsection{Subsystem anticoncentration and long-range MIE}

Below we show that random coarse-grained Clifford circuits exhibit subsystem anticoncentration in the sense described in the Introduction. In particular, we consider a tripartition $[n]=ABC$ of the qubits such that $|A|=1$, and the corresponding set $S$ defined in \cref{eq:defS1,eq:defS2}. We then give an upper bound on $\mathbb{E}_D[S]$ that approaches zero as $|C|,n\rightarrow \infty$. From our upper bound it follows directly that the family of random coarse-grained Clifford circuits satisfies the long-range MIE property. We use \cref{thm:anticonc} to infer that the family of random coarse-grained universal circuits also satisfies the long-range MIE property (\cref{prop:lrmie}).

To upper bound the expected size of $S$, we first bound the probability that a random Pauli operator becomes $Z$-type when it is evolved backwards in time through a random coarse-grained Clifford circuit. 

Given a Pauli operator $P$ defined on the qubits $[n]$, let $\supp(P) \subseteq [n]$ be set of qubits on which $P$ takes a value other than the identity. For any Pauli operator $P$ and random coarse-grained Clifford circuit $D$, it is relatively easy to see that the probability of $D^{\dagger} P D$ being $Z$-type is determined entirely by the overlap of $\supp(P)$ with the second layer of gates $W_1,\ldots, W_{m^2}$. To keep track of this overlap, for any Pauli $P$ we let the \emph{cluster} $\mathcal{G}(P)$ denote the set of all second layer gates with non-trivial overlap with $\supp(P)$. Such a cluster can be specified by filling in squares of an $m \times m$ grid. This grid representation leads naturally to the notion of the \emph{size} and \emph{perimeter} of a cluster $\mathcal{G}(P)$, which we define next. 

\begin{definition}
    Given any cluster $\mathcal{G}(P)$, identify it with a subset of the $m \times m$ grid in the natural way. Then let the \emph{size} of the cluster, denoted $\abs{\mathcal{G}(P)}$ be the number of grid squares contained in the cluster. Also let the \emph{perimeter} of the cluster, denoted $\mathrm{Per}(\mathcal{G}(P))$ be the number of grid edges on the boundary of the cluster, excluding edges running along the outside of the grid (see \cref{fig:cluster size and perimeter}). 
\end{definition}

\begin{figure}
\centering
\begin{tikzpicture}[every node/.style={minimum size=1cm-\pgflinewidth, outer sep=0pt}]
    \draw[step=1cm,color=black, thick] (0,0) grid (5,5);
    \node[fill=lightgray] at (0.5,4.5) {};
    \node[fill=lightgray] at (1.5,4.5) {};
    \node[fill=lightgray] at (2.5,4.5) {};
    \node[fill=lightgray] at (3.5,4.5) {};
    \node[fill=lightgray] at (0.5,3.5) {};
    \node[fill=lightgray] at (2.5,3.5) {};
    \node[fill=lightgray] at (4.5,3.5) {};
    \node[fill=lightgray] at (4.5,2.5) {};
    \node[fill=lightgray] at (4.5,1.5) {};
    \node[fill=lightgray] at (3.5,2.5) {};
    \node[fill=lightgray] at (1.5,1.5) {};

    \draw[teal, ultra thick] (0,3) -- (1,3) -- (1,4) -- (2,4) -- (2,3) -- (3,3) -- (3,4) -- (4,4) -- (4,5); 
    \draw[teal, ultra thick] (5,4) -- (4,4) -- (4,3) -- (3,3) -- (3,2) -- (4,2) -- (4,1) -- (5,1); 
    \draw[teal, ultra thick] (1,1) -- (1,2) -- (2,2) -- (2,1) -- (1,1);
\end{tikzpicture}
\caption{A cluster $\mathcal{G}(P)$ indicated in \textcolor{gray}{gray}, with size $\abs{\mathcal{G}(P)} = 11$ and \textcolor{teal}{perimeter} $\mathrm{Per}(\mathcal{G}(P)) = 19$. Note edges of the cluster running along the exterior of the grid do not contribute to the perimeter. }
\label{fig:cluster size and perimeter}
\end{figure}
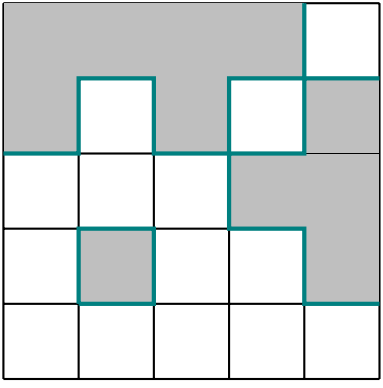

The following theorem bounds the probability that a Pauli is $Z$-type after conjugation by a random coarse-grained Clifford circuit.

\begin{lemma}
\label{lem:coarse_grained_prob_Ztype}
Let $P$ be a Pauli operator, and let $D$ be a random coarse-grained Clifford circuit. Also let $a = \abs{\mathcal{G}(P)}$ and $l = \mathrm{Per}(\mathcal{G}(P))$ denote the size and perimeter of $\mathcal{G}(P)$, respectively.  Then 
\begin{align}
    \Pr[D^{\dagger}  P D \text{ is $Z$-type}] \leq \left(\frac{1}{2} \right) ^{a \tau^2 + \frac{1}{4}l (\tau^2 / 12 - 1)}.
\end{align}
\end{lemma}

Now let us consider the measurement-induced entanglement scenario in the case where $A$ and $C$ are single qubits. We will use first-moment methods to establish the existence of measurement-induced entanglement in the coarse-grained architecture for $\tau=\Omega(\sqrt{\log(m)})$.

We first bound the expected number of binary strings $s \in S$ such that the perimeter of the cluster $\mathcal{G}(DZ(s)D^{\dagger})$ has a given length $l$.

\begin{lemma} \label{lem:expected_size_S_fixed_l}
    Let $D$ be a random coarse-grained Clifford circuit  and let $l$ be a positive integer. Let $A \subseteq[n]$ contain a single qubit and $C \subseteq [n]$ contain at least one qubit, with $A$ and $C$ disjoint. Then 
    \begin{align*}
    &\mathbb{E}_D\left[\abs{\{s \in S: \mathrm{Per}(\mathcal{G}(DZ(s)D^{\dagger})) = l\}} \right] 
    \nonumber \\
    &\hspace{100pt} \leq 2^{l(2\log(m)-\tau^2/48+5)}.
    \end{align*}
\end{lemma}

\begin{proof}
    Expanding out the definition of $S$ gives 
    \begin{align}
        &\abs{\{s \in S: \mathrm{Per}(\mathcal{G}(DZ(s)D^{\dagger})) = l\}} \nonumber \\
        &\hspace{10pt}=|\{s\in \{0,1\}^n: DZ(s)D^{\dagger}=(-1)^\alpha P_A\otimes Z(r)_B\otimes  I_C \nonumber\\ 
        &\hspace{24pt}\text{ for some } r\in \{0,1\}^{|B|}, \alpha\in \{0,1\}, P \in \{X,Y,Z\}\nonumber\\
        &\hspace{24pt}\textbf{ and } \mathrm{Per}(\mathcal{G}(DZ(s)D^{\dagger})) = l\}|.\nonumber
    \end{align}

    But for any fixed $r$ and $P$ we have
    \begin{align}
        &\exists s \in \{0,1\}^n : DZ(s)D^{\dagger}=(-1)^\alpha P_A\otimes Z(r)_B\otimes  I_C \nonumber \\
        &\hspace{40pt} \Leftrightarrow D^{\dagger} \left(P_A\otimes Z(r)_B\otimes  I_C\right) D \text{ is $Z$-type.}\nonumber
    \end{align}

    Then we also have 
    \begin{align}
        &\abs{\{s \in S: \mathrm{Per}(\mathcal{G}(DZ(s)D^{\dagger})) = l\}} \nonumber \\
        &=|\{(r,P): \mathrm{Per}(\mathcal{G}(P_A\otimes Z(r)_B\otimes  I_C)) = l, \textbf{ and } \nonumber \\
        &\hspace{30pt} D^{\dagger} \left(P_A\otimes Z(r)_B\otimes  I_C\right) D \text{ is $Z$-type}  \} |.
    \end{align}
    We can then bound the expected size of this set as follows. For a given cluster $\mathcal{G'}$, let $\Gamma(\mathcal{G'})$ denote the set of all pairs $r\in \{0,1\}^{|B|}$ and $P\in \{X,Y,Z\}$ such that $\mathcal{G}(P_A\otimes Z(r)_B\otimes  I_C) = \mathcal{G}'$. Then
    \begin{align}
    &\mathbb{E}_D\left[\abs{\{s \in S: \mathrm{Per}(\mathcal{G}(DZ(s)D^{\dagger})) = l\}}\right] \nonumber \\
    &= \sum_{\substack{\mathcal{G}' : \\  \mathrm{Per}(\mathcal{G}') = l}} \; \sum_{\substack{(r,P)\\\in \Gamma(\mathcal{G'})}} \Pr[D^{\dagger} \left(P_A\otimes Z(r)_B\otimes  I_C\right) D \text{ is $Z$-type.}] \nonumber \\
    &\leq \sum_{\substack{\mathcal{G}'  : \\ \mathrm{Per}(\mathcal{G}') = l}} \; \sum_{(r,P)\in \Gamma(\mathcal{G'})} 2^{-\abs{\mathcal{G}'} \tau^2 - \frac{1}{4}l(\tau^2/12 - 1)} \nonumber\\
    &\leq 3 \sum_{\substack{ \mathcal{G}'  : \\ \mathrm{Per}(\mathcal{G}') = l}} 2^{-\frac{1}{4}l(\tau^2/12 - 1)} \nonumber\\
    &\leq 6 \binom{2m(m-1)}{l} 2^{-\frac{1}{4}l(\tau^2/12 - 1)} \nonumber\\
    &\leq 2^{l( 2\log(m) - \tau^2/48 + 5)}.
    \end{align}
    where we used \cref{lem:coarse_grained_prob_Ztype} to go from the first line to the second and that there are at most $2^{\abs{\mathcal{G}'}\tau^2}$ strings $r \in \{0,1\}^n$  such that $\mathcal{G}(P_A\otimes Z(r)_B\otimes  I_C) = \mathcal{G}'$ to go from the second line to the third. To go from the third line to the fourth we count the number of clusters $\mathcal{G}'$ with perimeter $l$. Any such cluster can be specified by first choosing the location of the $l$ edges of the cluster, and then specifying the interior and exterior of the cluster. A crude upper bound gives that there are at most $ \binom{2m(m-1)}{l}$ ways to chose the $l$ edges of the cluster, and then two possible options for how to specify the interior and exterior of the cluster once the edges are fixed. This argument gives $\abs{\{\mathcal{G}'  : \mathrm{Per}(\mathcal{G}') = l \}} \leq 2 \binom{2m(m-1)}{l}$, which is the bound used above.
\end{proof}

Finally, we upper bound the expected size of $S$.

\begin{lemma}
\label{lem:S_bound_course_grained}
    In the setting of \cref{lem:expected_size_S_fixed_l} and choosing $\tau\geq\sqrt{1000\log_2(m)}$, we have
    \begin{align}
        \mathbb{E}_D[\abs{S}] \leq \frac{3}{2^{\abs{C} + 1}} + 2\cdot 2^{-\tau^2/100}.\label{eq:expected_size_of_S}
    \end{align}
\label{thm:expected_size_of_S}
\end{lemma}

\begin{proof}
We first observe 
\begin{align}
    \mathbb{E}_D[\abs{S}] \leq \sum_{l=0}^{\infty} \mathbb{E}_D\left[\abs{\{s \in S: \mathrm{Per}(\mathcal{G}(DZ(s)D^{\dagger})) = l\}}\right].\nonumber
\end{align}
Then \cref{lem:expected_size_S_fixed_l} and our assumption on the size of $d$ gives
\begin{align*}
\sum_{l=1}^{\infty} \mathbb{E}_D[\abs{\{s \in S: &\mathrm{Per}(\mathcal{G}(DZ(s)D^{\dagger})) = l\}}] \\
&\leq  \sum_{l=1}^{\infty} 2^{l(2\log_2(m)-\tau^2/48+5)} \\
&=\sum_{l=1}^{\infty} 2^{l(2\log_2(m)-\tau^2/96+5)-l\tau^2/96}\\
&\leq  \sum_{l=1}^{\infty} 2^{l(-\tau^2/96)}\\
&\leq 2\cdot 2^{-\tau^2/100}
\end{align*}
where in the second-to-last line we used the fact that $2\log_2(m)-\tau^2/96+5)\leq 2\log_2(m)-10\log_2(m)+5\leq  0$.
It remains to show a bound when $l = 0$. In this case, by the same logic what was used in the proof of \cref{lem:expected_size_S_fixed_l}, we have
\begin{align}
&\mathbb{E}_D\left[\abs{\{s \in S: \mathrm{Per}(\mathcal{G}(DZ(s)D^{\dagger})) = 0\}}\right] \nonumber \\
&\hspace{20pt}=\mathbb{E}_D\big[ |\{ r\in\{0,1\}^{\abs{B}}, P \in \{X,Y,Z\} : \nonumber \\
&\hspace{40pt} \mathrm{Per}(\mathcal{G}(P_A\otimes Z(r)_B\otimes  I_C)) = 0 \textbf{ and } \nonumber \\
&\hspace{40pt} D^{\dagger} \left(P_A\otimes Z(r)_B\otimes  I_C\right) D \text{ is $Z$-type}  \} | \big].
\end{align}
But we have $\mathrm{Per}(\mathcal{G}(P_A\otimes Z(r)_B\otimes  I_C)) = 0$ only if  the cluster $\mathcal{G}(P_A\otimes Z(r)_B \otimes  I_C)$ is the entire $m \times m$ grid. This cluster has size $m^2$ and there are at most 
\begin{align}
    3 (2^{\abs{B}}) = 3 ( 2^{n - \abs{A} - \abs{C} } ) = 3 ( 2^{n - 1 - \abs{C}})\nonumber
\end{align}
Pauli strings satisfying this condition. Then a straightforward application of \cref{lem:coarse_grained_prob_Ztype} gives 
\begin{align}
&\mathbb{E}_D\left[\abs{\{s \in S: \mathrm{Per}(\mathcal{G}(DZ(s)D^{\dagger})) = 0\}}\right] \nonumber \\
&\hspace{40pt}\leq 3 ( 2^{n - 1 - \abs{C}}) (2^{- m^2 \tau^2}) = 3 (2^{-\abs{C} - 1})
\end{align}
and the result follows. 
\end{proof}

\begin{theorem} [Formal version of \cref{thm:mie_coarsegrained}, part 1]\label{thm:mie-coarsegrained-formal}
The family of random coarse-grained universal circuits with $\tau\geq \sqrt{1000\log_2(m)}$ and $n$ sufficiently large satisfies the long-range MIE property. The same statement holds for random coarse-grained Clifford circuits.
\end{theorem}
\begin{proof}
Let $[n]=ABC$ be a tripartition of the qubits such that $|A|=1$ and $B$ is a square shielding region with side length $L$, centred at $A$ (see \cref{fig:abc}). Here $L$ is chosen so that $|C|\neq 0$. Note that due to the 2D geometry, this implies $|C|=\Omega(\sqrt{n})$.

Let $U$ be a random coarse-grained universal circuit with $\tau\geq \sqrt{1000\log_2(m)}$. Plugging the upper bound from \cref{thm:expected_size_of_S} into \cref{thm:anticonc} we get 
\[
\mathbb{E}_U \mathbb{E}_{x} \mathrm{Tr}(\rho(x)^2)\leq \frac{1}{2}+O(1/\mathrm{poly}(m))=\frac{1}{2}+O(1/\mathrm{poly}(n)).
\]
Taking $n$ sufficiently large we can ensure the RHS is less than any constant $c\in (1/2,1)$.

If we replace $U$ with a random coarse-grained Clifford circuit, then we can use \cref{thm:expected_size_of_S} along with \cref{lem:post-measurement_criterion_finite} and Markov's inequality to arrive at the same conclusion. 
\end{proof}
\subsection{Long-range tripartite MIE}

We have shown that random coarse-grained universal or Clifford circuits have the long-range MIE property (\cref{prop:lrmie}), as a consequence of the subsystem anticoncentration expressed in \cref{thm:expected_size_of_S}.  Below we specialize to random coarse-grained Clifford circuits, and we show that \cref{thm:expected_size_of_S} also implies long-range tripartite MIE. As discussed in the Introduction, this gives rise to an unconditional quantum advantage.

 We first review some facts about entanglement in stabilizer states.  For any graph $G = (V,E)$, the $\abs{V}$-vertex graph state $\ket{\Phi_G}$ is defined to be the stabilizer state with stabilizer group generated by 
\begin{align}
    g_v = X_v \left(\prod_{u : (u,v) \in E} Z_u\right).
\end{align}
It is known that any (finite dimensional) stabilizer state is locally Clifford equivalent to a graph state~\cite{van2004graphical}. 

The following lemma shows that GHZ-type entanglement can always be produced by making computational basis measurements on a connected graph state with sufficiently many vertices.

\begin{lemma}
\label{lem:producing_GHZ_type_entanglement}
Let $G = (V,E)$ be a connected graph with at least 3 vertices and let $\ket{\Phi_G}$ be the associated graph state. Then there exist indices $h,i,j \in V$ such that, for any $x \in \{0,1\}^{\abs{V} - 3}$, the state 
\begin{align}
    \ket{\Phi_{h,i,j}} 
    = \bra{x}_{V \backslash \{h,i,j\}} \ket{\Phi_G}
\end{align}
is locally Clifford equivalent to a GHZ state (up to normalization).
\end{lemma}
\begin{proof}
The post-measurement state after measuring a qubit $v$ of a graph state $\Phi_G$ in the computational basis and obtaining outcome $y\in \{0,1\}$ is
\[
\bigg(\prod_{w:\{v,w\}\in E(G)} Z^y_w\bigg)|\Phi_{G(V\setminus\{v\})}\rangle 
\]
where $G(V\setminus\{v\})$ is the induced subgraph of $G$ on all $V\setminus\{v\}$,  obtained from $G$ by removing vertex $v$ and all edges incident to it. Starting from a connected graph $G$ with at least three vertices, we choose $3$ vertices $\Omega=\{h,i,j\}$ such that the subgraph $G(\Omega)$ of $G$ induced by $\Omega$ is a connected three-vertex graph. Using the above fact about computational basis measurement, we see that after measuring all vertices in $V\setminus\{h,i,j\}$ the postmeasurement state is locally Clifford equivalent to the graph state $|\Phi_{G(\Omega)}\rangle$ where $G(\Omega)$ is a connected three-vertex graph. Finally, we use the well-known fact that any 3 vertex connected graph state is locally Clifford equivalent to a GHZ state (this can be confirmed by a direct calculation).
\end{proof}

Now consider a random coarse-grained Clifford circuit with $\tau\geq \sqrt{1000\log_2(m)}$ as in \cref{thm:expected_size_of_S}. In this setting, a straightforward consequence of \cref{lem:post-measurement_criterion_finite,lem:S_bound_course_grained} is that measuring all but two output qubits of the randomly chosen circuit in the computational basis produces a two-qubit entangled stabilizer state on the unmeasured qubits with probability at least $1/4 - \exp(-\Omega(\tau^2))$. The next theorem shows that tripartite GHZ-type entanglement is also ubiquitous.

\begin{theorem}[Formal version of \cref{thm:mie_coarsegrained}, part 2]
Random coarse-grained Clifford circuits with $\tau\geq \sqrt{1000\log_2(m)}$ satisfy the long-range tripartite MIE property (\cref{prop:ghz}). In particular, let $D$ be an $n$-qubit random coarse-grained Clifford circuit, with $\tau\geq \sqrt{1000\log_2(m)}$. Let $\{h,i,j\}$ be a uniformly random triple of qubits. Then $\{h,i,j\}$ exhibit GHZ-type MIE with probability at least $0.005$ over the random choice of $D$ and the random choice of $\{h,i,j\}$.
\end{theorem}
\begin{proof}
Let $D$ be an $n$-qubit random coarse-grained Clifford circuit with $\tau\geq \sqrt{1000\log_2(m)}$.

Let us say that $\{h,i,j\}\subseteq [n]$ is a \textit{good triple}, iff qubits  $\{h,i,j\}$ share GHZ-type entanglement after all other qubits have been measured with probability at least $0.05$ over the random choice of $D$.

Below we prove the following: for any set of qubits $V\subseteq [n]$ with $\abs{V} = 5$, there exists a good triple $\{h,i,j\}\subseteq V$.

The theorem then follows from this statement. To see this, consider the following procedure. First, choose a uniformly random subset of qubits $V\subseteq [n]$ of size $|V|=5$, and then choose a uniformly random triple $\{h,i,j\}\subseteq V$. By the above, the probability that we are lucky and choose a good triple is at least $1/\binom{5}{3}=1/10$.

Moreover, clearly this procedure samples a triple $\{h,i,j\}$ from the uniform distribution over all $\binom{n}{3}$ triples. 

It remains to show that for any set of qubits $V\subseteq [n]$ with $\abs{V} = 5$, there exists a good triple $\{h,i,j\}\subseteq V$.

We first consider the state of the qubits in $
V$ after the circuit $D$ is applied and all other qubits are measured in the computational basis. For any qubit $v \in V$, \cref{lem:post-measurement_criterion_finite,lem:S_bound_course_grained} give that there is no post-measurement entanglement between qubit $v$ and qubits in the set $V\backslash \{v\}$ with probability at most 
\[
3/2^5 + 2\cdot 2^{-\tau^2/100}\leq 3/2^5 + 2\cdot 2^{-10}< 0.1.
\]
Then the union bound gives that there exists a qubit $v' \in V$ with no post-measurement entanglement between $v'$ and the qubits $V\backslash \{v'\}$ with probability at most 
\begin{align}
\abs{V}\cdot 0.1 < 0.5.\nonumber
\end{align}
Equivalently, there is post measurement entanglement between each qubit $v' \in V$ and qubits $V\backslash\{v'\}$ with probability at least $1/2$. But then the post measurement state of the qubits $V$ is locally Clifford equivalent to a graph state $\ket{\Phi_G}$ corresponding to a 5 vertex graph with no isolated vertices (this is because there is no entanglement between qubits corresponding to disconnected vertices of a graph state). It follows that the graph $G$ must have either a 3 or 5 vertex connected component. If $G$ has a 3 vertex connected component the qubits corresponding to this component are in a state locally Clifford equivalent to a GHZ state, and we are done. Otherwise, the graph state has a 5 vertex connected component. But then by \cref{lem:producing_GHZ_type_entanglement} there are vertices $h,i,j \in V$ such that, for any $x \in \{0,1\}^2$ we have that
\begin{align}
    \ket{\Phi_{hij}} = \bra{x}_{V\backslash \{h,i,j\}} \ket{\Phi_G} \label{eq:corse_grained_GHZ_1}
\end{align}
is locally Clifford equivalent to a GHZ state. Now let $\ket{\Psi_V}$ denote the post measurement state of the qubits in the set $V$. Since this state is locally Clifford equivalent to $\ket{\Phi_G}$ we have 
\begin{align}
    \ket{\Psi_{V}} = \bigotimes_{v \in V} C_v \ket{\Phi_G}\nonumber
\end{align}
where each $C_v$ is a $1$-qubit Clifford unitary. Furthermore, since the distribution of coarse-grained Clifford circuits is unchanged under local Clifford operations, we can assume the Clifford unitaries $C_v$ are chosen uniformly at random. 
Then the post measurement state of qubits $h,i,j$ when the remaining two qubits of $V$ are measured in the computational basis and outcome $x\in \{0,1\}^2$ is observed is given by
\begin{align}
    \ket{\Psi_{h,i,j}} &= \bra{x}_{V\backslash \{h,i,j\}} \bigotimes_{v \in V} C_v \ket{\Phi_G}. 
\end{align}
Comparing with \cref{eq:corse_grained_GHZ_1} we see this is locally Clifford equivalent to a GHZ state provided the state 
\begin{align}
\bigotimes_{v \in V\backslash \{h,i,j\}} C_v^{\dagger}  \ket{x}_{V\backslash \{h,i,j\}}\nonumber
\end{align}
is a computational basis state. But since we can take the $C_v$ to be random Cliffords this occurs with probability exactly $1/9$. Putting this all together we see that $\{h,i,j\}$ is a good triple: the state $\ket{\Psi_{h,i,j}}$ is locally Clifford to a GHZ state with probability at least 
$\frac{1}{9}\cdot \frac{1}{2}\geq 0.05$
over the random choice of $D$.
\end{proof}

\subsection{Approximate implementation of coarse-grained circuits with two-qubit gates}

Lastly, we demonstrate that coarse-grained circuits can be approximately implemented using two-qubit gates, instead of gates that act on blocks of $\tau \times \tau$ qubits.
We shall see that this yields a family of $O(\log(n))$-depth quantum circuits composed of Haar-random two-qubit gates that satisfy the long-range MIE property.

Our construction replaces each $\tau \times \tau$ Haar random gate in a random coarse-grained universal circuit by an $\varepsilon$-approximate unitary 2-design for a sufficiently small approximation error $\varepsilon=2^{-\Omega(\tau^2)}$.
Such unitary $2$-designs can, for instance, be implemented using a 1D brickwork quantum circuit of depth $O(\tau^2)$ composed of Haar random two-qubit gates which is arranged in a snake-like pattern (or more formally, along a Hamiltonian path) on the square of $\tau \times \tau$ qubits \cite{brandao2016local,hunter2019unitary}.
This results in a distribution of $n$-qubit random quantum circuits with an overall depth of $O(\tau^2)$. We denote these compiled random coarse-grained circuits by $\tilde{\mathcal{U}}$.
We can analogously define $\tilde{\mathcal{D}}$ to be the family of random Clifford circuits obtained by replacing the Haar random two-qubit gates in  $\tilde{\mathcal{U}}$ by random two-qubit Clifford gates.
\begin{theorem} [Long-range MIE in compiled coarse-grained circuits with $2$-qubit gates]\label{thm:lrmie-compiled}
The family of complied random coarse-grained circuits $\tilde{\mathcal{U}}$ satisfies the long-range MIE property when $\tau\geq \sqrt{1000\log_2(m)}$ and $n$ sufficiently large. The same statement holds for compiled random coarse-grained Clifford circuits $\tilde{\mathcal{D}}$.
\end{theorem}
\begin{proof}
We begin by applying \Cref{thm:anticonc}, which allows us to reduce the analysis of the long-range MIE property in $\tilde{\mathcal{U}}$ to that in $\tilde{\mathcal{D}}$. 
To establish the long-range MIE property in $\tilde{\mathcal{D}}$, we focus on bounding the expected size of the set$~S$. 
In the Appendix, we demonstrate that the family of random Clifford circuits $\tilde{\mathcal{D}}$ satisfies the bound in \cref{eq:expected_size_of_S} of \Cref{lem:S_bound_course_grained}. 
We can then employ the same reasoning as in the proof of \Cref{thm:mie-coarsegrained-formal} to establish the long-range MIE property for circuits drawn from $\tilde{\mathcal{U}}$, as stated in \Cref{prop:lrmie}.
\end{proof}

\vspace{1em}
\emph{Note added:} As we finalized this paper, we became aware of concurrent works \cite{schuster2024random, laracuente2024approximate} that also study the 2D coarse-grained architecture introduced here. While we focus on establishing quantum advantage,  long-range MIE, and anti-concentration properties, these works demonstrate that with a suitable choice of $\log(n)$-sized coarse-grained gates, a similar architecture can also exhibit the unitary $k$-design property. 

\section{Acknowledgments}
DG acknowledges the support of
 the Natural Sciences and Engineering Research Council of Canada through grant number RGPIN-2019-04198. DG is a fellow of the Canadian Institute for
Advanced Research, in the quantum information science program. Research at Perimeter
Institute is supported in part by the Government of Canada through the Department of
Innovation, Science and Economic Development Canada and by the Province of Ontario
through the Ministry of Colleges and Universities.
MS is supported by AWS Quantum Postdoctoral Scholarship and funding from the National Science Foundation NSF CAREER award CCF-2048204.
Institute for Quantum Information and Matter is an NSF Physics Frontiers Center. 

\bibliographystyle{unsrt}
\bibliography{main}

\appendix

\onecolumngrid

\section{Analysis of \cref{alg:sim}}
Here we prove \cref{thm:simalg}. The proof is obtained by combining the fidelity bound from \cref{lem:purity} with the robustness property of the gate-by-gate algorithm established in Ref. \cite{bravyi2022simulate}. Our setting differs slightly from the one considered there (this is due to the fact that $\tilde{\rho}_t(\omega_B)$ is not a pure state in general); below we adapt the proof to our setting.   

The \textit{gate-by-gate} algorithm of Ref.~\cite{bravyi2022simulate} is obtained from \cref{alg:sim} by making a small but consequential replacement: replacing $\tilde{q}_t(y|x_{B(t)})$ with $q_t(y|x_{B(t)}x_{C(t)})$ in line 9. We shall analyze \cref{alg:sim} by comparing it to the gate-by-gate algorithm and using the performance guarantee of the latter algorithm. Specifically, in Ref.~\cite{bravyi2022simulate} it is shown that the binary string $x$ at the end of the $t$th iteration of the for loop of the gate-by-gate algorithm is distributed according to 
\[
P_t(x)\equiv |\langle x|U_tU_{t-1}\ldots U_1|0^n\rangle|^2 \qquad x\in \{0,1\}^n.
\]

Write $W_t(x)$ for the probability of sampling $x\in \{0,1\}^n$ at the end of the $t$th iteration of the for loop of \cref{alg:sim}.

First suppose that $t\in \Gamma$, so that $U_t$ is a single-qubit gate. Let $A=A(t), B=B(t)$, and $C=C(t)$. Then
\begin{align}
W_t(x)&=\sum_{y:y_{BC}=x_{BC}} W_{t-1}(y) \tilde{q}_t(x_A|x_B)\nonumber\\
&=\sum_{y:y_{BC}=x_{BC}} \left(W_{t-1}(y)-P_{t-1}(y)\right) \tilde{q}_t(x_A|x_B)+\sum_{y:y_{BC}=x_{BC}}P_{t-1}(y) \tilde{q}_t(x_A|x_B)\nonumber\\
&=\sum_{y:y_{BC}=x_{BC}} \left(W_{t-1}(y)-P_{t-1}(y)\right) \tilde{q}_t(x_A|x_B)+\sum_{y:y_{BC}=x_{BC}}P_{t-1}(y) \left(\tilde{q}_t(x_A|x_B)-q_t(x_A|x_Bx_C)\right)+P_t(x),
\label{eq:differences}
\end{align}
where we used the fact that $P_t(x)=\sum_{y:y_{BC}=x_{BC}} P_{t-1}(y)q_t(x_A|x_Bx_C)$. Now using \cref{eq:differences}, summing over $x\in \{0,1\}^{n}$, and using the triangle inequality, we get
\begin{align}
\|W_t-P_t\|_1&\leq \sum_{x_{BC}\in \{0,1\}^{BC}}\sum_{y:y_{BC}=x_{BC}} |\left(W_{t-1}(y)-P_{t-1}(y)\right)|\sum_{x_A\in \{0,1\}^{|A|}}\tilde{q}_t(x_A|x_B)\nonumber\\
&+\sum_{x_{BC}\in \{0,1\}^{BC}}\sum_{y:y_{BC}=x_{BC}}P_{t-1}(y) \sum_{x_{A}\in \{0,1\}^{|A|}}\left|\tilde{q}_t(x_A|x_B)-q_t(x_A|x_Bx_C)\right|\nonumber\\
 &= \|W_{t-1}-P_{t-1}\|_1+\mathbb{E}_{z\sim P_{t-1,BC}}\left(\|q_t(\cdot|z_Bz_C) -\tilde{q}_t(\cdot |z_B)\|_1\right).
\label{eq:wtpt}
\end{align}

Using \cref{eq:q,eq:qtil} and the relationship between trace distance and quantum fidelity, we have
\begin{equation}
\|q_t(\cdot|z_Bz_C) -\tilde{q}_t(\cdot |z_B)\|_1\leq \| \rho_t(z_Bz_C)-\tilde{\rho}_t(z_B)\|_1\leq  2\sqrt{1-F(\rho_t(z_Bz_C), \tilde{\rho}_t(z_B))}.
\label{eq:tracedist}
\end{equation}
Plugging \cref{eq:tracedist} into \cref{eq:wtpt} and using the fact that $\mathbb{E}[\sqrt{X}]\leq (\mathbb{E}[X])^{1/2}$ for any non-negative random variable $X$ gives
\begin{align}
\|W_t-P_t\|_1&\leq  \|W_{t-1}-P_{t-1}\|_1+2\left(\mathbb{E}_{z\sim P_{t-1,BC}} \left[1-F(\rho_t(z_Bz_C), \tilde{\rho}_t(z_B))\right]\right)^{1/2}\nonumber\\\
&\leq \|W_{t-1}-P_{t-1}\|_1+2\left(1-\mathbb{E}_{z_B\sim P_{t-1,B} \left[\mathrm{Tr}(\tilde{\rho}_t(z_B)^2)\right]}\right)^{1/2} \qquad t\in S.
\label{eq:ts}
\end{align}
where we used \cref{lem:purity} and the fact that $P_{t-1,BC}=P_{t,BC}$ (since the unitary gate $U_t$ acts only on qubit $A$ and does not change the marginal distribution on $BC$). 

Next suppose $t\notin \Gamma$, so $U_t=\mathrm{CNOT}_{ij}$ for some pair of qubits $i\neq j$. Let $f_t(x)$ be the function applied to $x$ in line $4$ of \cref{alg:sim}; that is, it updates $x$ by replacing $x_j\leftarrow x_i\oplus x_j$. Then
\[
W_t(x)=W_{t-1}(f_t(x)) \quad \text{ and } \quad P_t(x)=P_{t-1}(f_t(x)),
\]
so in this case 
\begin{equation}
\|W_t-P_t\|_1=\|W_{t-1}-P_{t-1}\|_1 \qquad t\notin \Gamma.
\label{eq:tnots}
\end{equation}

Applying \cref{eq:ts,eq:tnots} for $t=1,2,\ldots, m$ and using the fact that $P_0=W_0$, we arrive at
\[
\|W_m-P_m\|_1\leq 2\sum_{t\in \Gamma} \left(1-\Exp_{z\sim P_{t-1,B}} \left[\mathrm{Tr}(\tilde{\rho}_t(z)^2)\right]\right)^{1/2}=2\sum_{t\in \Gamma} \left(1-\Exp_{z\sim P_{t,B}} \left[\mathrm{Tr}(\tilde{\rho}_t(z)^2)\right]\right)^{1/2}
\]
where in the last equality we used the fact that $U_t$ is a single-qubit gate acting on qubit $A$, so $P_{t-1,B}=P_{t,B}$. This completes the proof of \cref{thm:simalg}.

\section{Proof of \cref{lem:coarse_grained_prob_Ztype}}\label{app:compiledprob_ZType}
\noindent Before proving this theorem we require a preliminary result about the effect of conjugating Pauli operators through random Clifford circuits. 

\begin{lemma}

\label{lem:random Paulis and Cliffords}
Let $P_1, P_2,\ldots, P_r$ be uniformly random non-identity Pauli strings of various lengths, chosen so that $P_1 \otimes P_2 \otimes \cdots \otimes P_r$ defines a Pauli string of length $n$. Denote this Pauli string $P_{[n]}$ and for any $Q \subseteq [n]$ let $P_Q$ denote the corresponding truncation of this string. Let $C$ be a uniformly random $t$-qubit Clifford operator. Then
\begin{enumerate}[label = (\alph*), ref = (\alph*)]
    \item \label{item:random Pauli and Clifford fully contained Pauli} For any $Q \subseteq [n]$ with $\abs{Q} = t$: 
    \begin{align}
        \Pr\left[ C^\dagger P_Q C \text{ is $Z$-type}\right] \leq \left(\frac{1}{2}\right)^t.
    \end{align}
    \item \label{item:random Pauli and Clifford partially contained Pauli}
    For any $Q \subseteq [n]$ with $\abs{Q} = k < t$:
    \begin{align}
    \Pr\left[ C^\dagger \left( P_Q \otimes I ^{\otimes t-k}\right) C \text{ is $Z$-type}\right] \leq \left(\frac{1}{4}\right)^k + \left(\frac{1}{2}\right)^t.
    \end{align}
\end{enumerate}
Moreover, these bounds still hold even after arbitrary conditioning on the value of Paulis in the complimentary string $P_{[n] \backslash Q}$. 
\end{lemma}

\begin{proof}
We first argue that for any $Q$ the probability of $P_Q$ being identity is at most $4^{-\abs{Q}}$, even after conditioning on an arbitrary value for $P_{[n] \backslash Q}$. If the Pauli's $P_1, P_2,\ldots, P_r$ were uniformly random (including the possibility of being the identity) then each individual Pauli in $P_{[n]}$ would be chosen uniformly at random and this statement would be immediate, with the bound $4^{-\abs{Q}}$ holding exactly. That is, if we let $P_1', P_2', \ldots, P_r'$ be uniformly random Pauli strings which may be the identity we have 
\begin{align}
\Pr\left[\left(P_1' \otimes P_2' \otimes \cdots \otimes P_r'\right)_Q = I \; | \; \left(P_1' \otimes P_2' \otimes \cdots \otimes P_r'\right)_{[n] \backslash Q} = K' \right] = 4^{-\abs{Q}}
\end{align}
for any choice of $Q$ and Pauli string $K'$. But now we observe that further enforcing that each $P_i'$ is non-identity in the above equation can only increase the probability that $\left(P_1' \otimes P_2' \otimes \cdots \otimes P_r'\right)_Q$ is non-identity since, by Bayes' rule,
\begin{align*}
    &\frac{\Pr\left[\left(P_1' \otimes P_2' \otimes \cdots \otimes P_r'\right)_Q = I \; | \; \left(P_1' \otimes P_2' \otimes \cdots \otimes P_r'\right)_{[n] \backslash Q} = K' \textbf{ and } P_1', P_2', \ldots , P_r' \neq I \right]}{\Pr\left[\left(P_1' \otimes P_2' \otimes \cdots \otimes P_r'\right)_Q = I \; | \; \left(P_1' \otimes P_2' \otimes \cdots \otimes P_r'\right)_{[n] \backslash Q} = K' \right]} \\
    &\hspace{30pt}= \frac{\Pr\left[P_1', P_2', \ldots , P_r' \neq I  \; | \; \left(P_1' \otimes P_2' \otimes \cdots \otimes P_r'\right)_{[n] \backslash Q} = K' \textbf{ and } \left(P_1' \otimes P_2' \otimes \cdots \otimes P_r'\right)_Q = I\right]}{\Pr\left[P_1', P_2', \ldots , P_r' \neq I  \; | \; \left(P_1' \otimes P_2' \otimes \cdots \otimes P_r'\right)_{[n] \backslash Q} = K' \right]}
\end{align*}
and the quantity on the on the right hand side of the above equation is clearly bounded above by one. The statement follows.

A similar logic can be applied to prove statements $(a)$ and $(b)$ directly but, for concreteness, we will take a slightly more involved route. We start with $(a)$. The random Clifford $C$ will map a non-identity Pauli to a uniformly random non-identity Pauli on $t$ qubits. There are $4^t - 1$ such Paulis and $2^t - 1$ of them are $Z$-type, so
\begin{align}
    \Pr\left[ C^\dagger P_Q C \text{ is $Z$-type}\right] &= \Pr[P_Q \text{ is identity}] + \Pr[P_Q \text{ is non-identity}] \left(\frac{2^t - 1}{4^t - 1} \right) \nonumber\\
    &\leq 4^{-t} + (1 - 4^{-t})\left(\frac{2^t - 1}{4^t - 1} \right) \nonumber\\
    &= \frac{1}{4^t} + \frac{2^t - 1}{4^t} = \frac{1}{2^t},
\end{align}
where we used our upper bound on the probability of $P_Q$ begin the identity along with the fact that $\frac{2^t - 1}{4^t - 1} \leq 1$ to go from the first line to the second. 

To prove $(b)$ we begin with a similar expression and find
\begin{align}
    \Pr\left[ C^\dagger\left(P_Q\otimes I^{\otimes t-k}\right)C \text{ is $Z$-type}\right] &= \Pr[P_Q \text{ is identity}] + \Pr[P_Q \text{ is non-identity}] \left(\frac{2^t - 1}{4^t - 1} \right) \nonumber\\
    &\leq 4^{-k} + (1 - 4^{-k})\left(\frac{2^t - 1}{4^t - 1} \right) \nonumber\\
    &\leq 4^{-k} + \left(\frac{2^t - 1}{4^t - 1} \right) \leq 4^{-k} + 2^{-t},
\end{align}
as desired. 
\end{proof}

We now proceed to the proof of \cref{lem:coarse_grained_prob_Ztype}.

\begin{proof}[Proof (\cref{lem:coarse_grained_prob_Ztype})]
We write $D = W V$ where $V$ and $W$ denote the first and second layer of random Clifford circuits, respectively. We first consider the Pauli string obtained by conjugating $P$ by just the second layer of Cliffords, which we can write as $P'=W^\dagger P W$. This string consists of a random non-identity Pauli string acting on every $\tau \times \tau$ patch of qubits in the support of $\mathcal{G}(P)$, and is identity elsewhere. There are $a$ such patches by definition, and so $a \tau^2$ qubits on which this random Pauli string can potentially be non-identity. Let $Q'$ denote the set of all these qubits.

Now we consider the first layer of random Clifford unitaries. Recall from the definition of random coarse-grained circuits that we have $V = V_1 \otimes V_2 \otimes \cdots \otimes V_{(m-1)^2}$, where each $V_i$ is a random $\tau \times \tau$ Clifford. For any $1\leq i\leq (m-1)^2$ let $Q'_i = Q' \cap \supp(V_i)$ and let $P'_i$ be the restriction of the Pauli string $P'$ to the support of $V_i$. Finally let $Q_0' = Q' \backslash (\cup_i Q'_i)$ and $P'_0$ be the restriction of $P'$ to $Q'_0$. Also, as a notational convenience, let $V_0$ denote the identity gate acting on the qubits in $Q_0'$. (This is to accommodate qubits in $Q'$ which lie on the boundary of the circuit, and thus are not in the support of any $V_i$ for $i \geq 1$). We want to bound the quantity
\begin{align}
    \Pr[D^\dagger P D \text{ is $Z$-type}] &= \Pr\left[ V_i^\dagger P'_i V_i \text{ is $Z$-type for all $0 \leq i \leq (m-1)^2$}\right]\nonumber \\
    &= \prod_{i = 0}^{(m-1)^2} \Pr\left[ V_i^\dagger P'_i V_i \text{ is $Z$-type} \; | \; V_j^\dagger P'_j V_j \text{ is $Z$-type for all $0 \leq j < i $}\right]. \label{eq:zprob first layer decomp}
\end{align}
It is immediate from definitions that 
\begin{align}
\Pr[V_0^{\dagger} P_0' V_0 \text{ is $Z$-type}] = \Pr[P_0' \text{ is $Z$-type}] \leq 2^{-\abs{Q_0'}}.
\label{eq:V_0}
\end{align}
Then let $\mathcal{V}_\text{int}$ be the set of all $V_i$ with $0 < i \leq (m-1)^2$ whose support is contained entirely in $Q'$, and $\mathcal{V}_\text{perim}$ be the set of all $V_i$ ith $0 < i \leq (m-1)^2$ whose support overlaps partially with $Q'$. 
As an immediate consequence of \cref{lem:random Paulis and Cliffords}~\labelcref{item:random Pauli and Clifford fully contained Pauli} we see that for any $V_i \in \mathcal{V}_\text{int}$
\begin{align}
\Pr\left[ V_i^\dagger P'_i V_i \text{ is $Z$-type} \; | \; V_j^\dagger P'_j V_j \text{ is $Z$-type for all $j < i $}\right] \leq 2^{-\tau^2} = 2^{- \abs{Q_i'}}.
\label{eq:V_int}
\end{align}
Similarly, for any $V_i \in \mathcal{V}_\text{perim}$, \cref{lem:random Paulis and Cliffords}~\labelcref{item:random Pauli and Clifford partially contained Pauli} gives that 
\begin{align}
\Pr\left[ V_i^\dagger P'_i V_i \text{ is $Z$-type} \; | \; V_j^\dagger P'_j V_j \text{ is $Z$-type for all $j < i $}\right] \leq 4^{-\abs{Q_i'}} + 2^{-\tau^2}.
\end{align}
For any $V_i \in \mathcal{V}_\text{perim}$ we have that at least $1/4$ and at most $3/4$'s of the qubits in the support of $V_i$ overlap with $Q'$, and hence $\tau^2/4 \leq \abs{Q_i'} \leq 3\tau^2/4$. Then we can also write 
\begin{align}
    \Pr\left[ V_i^\dagger P'_i V_i \text{ is $Z$-type} \; | \; V_j^\dagger P'_j V_j \text{ is $Z$-type for all $j < i $}\right] &\leq 4^{-\abs{Q_i'}} + 2^{-4 \abs{Q_i'} / 3}\nonumber \\
    &= 2^{-\abs{Q_i'}} \left(2^{-\abs{Q_i'}} + 2^{-\abs{Q_i'}/3} \right)\nonumber \\
    &\leq 2^{-\abs{Q_i'} - (\tau^2/12) + 1}.
\label{eq:V_perim}
\end{align}
Combining \cref{eq:V_0,eq:V_int,eq:V_perim} gives 
\begin{align}
\Pr[D^\dagger PD\text{ is $Z$-type}]&\leq 2^{-\abs{Q_0'}} \left(\prod_{V_i \in \mathcal{V}_\text{int}} 2^{-\abs{Q_i'}}\right)\left(\prod_{V_j \in \mathcal{V}_\text{perim}} 2^{-\abs{Q_j} - \tau^2/12 + 1}\right)\nonumber\\
&= 2^{- \abs{\mathcal{V}_\text{perim}} (\tau^2/12 - 1) - \sum_{i=0}^{(m-1)^2} \abs{Q'_i}}\nonumber \\
&= 2^{- \abs{\mathcal{V}_\text{perim}} (\tau^2/12 - 1) - \abs{Q'}}.
\end{align}
We observed previously that $\abs{Q'} = a\tau^2$. Every gate $V_i\in\mathcal{V}_{\text{perim}}$ intersects with four edges, and in the worst case, all four of the edges could be perimeter edges. At the same time, every perimeter edge intersects with at least one gate in $\mathcal{V}_\text{perim}$. Therefore, $|\mathcal{V}_\text{perim}|\geq\frac{1}{4}l$. Inserting those two values into the equation above completes the proof. 

\end{proof}

\section{Compiling the coarse-grained architecture to random quantum circuits}
\label{sec:compiling}

Here we consider the setup of \Cref{thm:lrmie-compiled} which involves compiled random coarse-grained circuits $\tilde{\mathcal{U}}$ and their Clifford versions $\tilde{\mathcal{D}}$ composed of two-qubit gates. 
We show that an equivalent version of \cref{lem:coarse_grained_prob_Ztype} holds for $\tilde{\mathcal{D}}$ with sufficiently small $\varepsilon$. 
The gates in the non-compiled coarse-grained circuits act on blocks of $\tau \times \tau$ qubits. 
Define $t=\tau^2$. Let $\mathcal{D}_t$ denote the uniform distribution over $t$-qubit Clifford operators. 
Let $\tilde{\mathcal{D}}_t$ be a distribution of $t$-qubit random Clifford circuits forming an $\varepsilon$-approximate unitary $2$-design w.r.t the operator norm. The approximate $2$-design property of $\tilde{\mathcal{D}}_t$ in this notion implies that for every $2^t$-by-$2^t$ matrix $A,B$ with $\lVert A\rVert_F=\lVert B\rVert_F=1$, it holds that
\begin{equation}\left|\Exp_{\tilde{C}\sim\tilde{\mathcal{D}}_t}\left[\tr(B\tilde{C}^\dagger A\tilde{C})^2\right]-\Exp_{C\sim\mathcal{D}_t}\left[\tr(BC^\dagger AC)^2\right]\right|\leq\varepsilon.
\label{eq:2_design}
\end{equation}
While the compilation is agnostic to any specific construction of $\varepsilon$-approximate $2$-design, for concreteness, we could implement the circuits in $\tilde{\mathcal{D}}_t$ using 1D brickwork random Clifford circuits with depth $O(\log(1/\varepsilon))$ \cite{brandao2016local,hunter2019unitary}. Alternatively, we can consider a local random quantum circuit architecture with $O(t\log(1/\varepsilon))$ gates where at each time step, a random $2$-qubit Clifford gate is applied to a randomly chosen qubit and one of its neighbours on the 2D grid \cite{mittal2023local}.
We note that the constructions in \cite{brandao2016local,hunter2019unitary,mittal2023local} are known to from an $\varepsilon$-approximate $2$-design if the circuits are composed of Haar random two-qubit gates or random gates drawn from a universal gate set.
However, we observe that the same constructions yield $\varepsilon$-approximate $2$-design if the gates are chosen randomly from $2$-qubit Clifford gates. 
This is the case because the $2$-design bounds in \cite{brandao2016local,hunter2019unitary,mittal2023local} are established by analyzing mathematical objects defined solely using the second moment operator associated with the $2$-qubit Haar measure $\mathcal{U}_2$
\begin{equation}
\Exp_{U\sim\mathcal{U}_2}\left[U^{\otimes 2}\otimes \left(\overline{U}\right)^{\otimes 2}\right],
\label{eq:moment_operator}
\end{equation}
and the architecture template---for example by lower bounding the spectral gap of a local Hamiltonian defined in terms of \cref{eq:moment_operator}. Since the uniform distribution over $2$-qubit Clifford gates $\mathcal{D}_2$ forms an exact unitary $2$-design,
$$\Exp_{C\sim\mathcal{D}_2}\left[C^{\otimes 2}\otimes \left(\overline{C}\right)^{\otimes 2}\right]=\Exp_{U\sim\mathcal{U}_2}\left[U^{\otimes 2}\otimes \left(\overline{U}\right)^{\otimes 2}\right].$$

Mirroring the proof strategy for \cref{lem:coarse_grained_prob_Ztype}, we begin by establishing a direct analog of \cref{lem:random Paulis and Cliffords} in \Cref{app:compiledprob_ZType}.

\begin{lemma}
Let $K$ be an arbitrary non-identity $t$-qubit Pauli operator, draw a $\tilde{C}\sim\tilde{\mathcal{D}}_t$, and define $P=\tilde{C}^\dagger K\tilde{C}$. Let $h,s\geq 1$ satisfy $h+s=t$. Decompose $P$ as $P=\pm P_1\otimes P_2$ such that $P_1$ and $P_2$ are $h$- and $s$-qubit Pauli operators respectively. Let $\emptyset\neq \mathcal{P}\subseteq\{I,X,Y,Z\}^{\otimes s}$. Then for every $\varepsilon\leq \frac{1}{3(4^t-1)}$, it holds that
\begin{equation}
\Pr(P_1=I^{\otimes h}|P_2\in\mathcal{P})\leq\frac{1}{4^h}(1+3(4^t-1)\varepsilon).
\end{equation}
\end{lemma}

\begin{proof}
In the following, we make use of the fact that for every non-identity Pauli $W\in\{I,X,Y,Z\}^{\otimes t}\setminus\{I^{\otimes t}\}$,
\begin{equation}
\Pr(P=\pm W)=\frac{1}{4^t}\Exp_{\tilde{C}\sim\tilde{\mathcal{D}}_t}\left[\tr(W \tilde{C}^\dagger K \tilde{C})^2\right].
\end{equation}
Thus, by \cref{eq:2_design},
\begin{equation}
\frac{1}{4^t-1}-\varepsilon\leq \Pr(P=\pm W)\leq \frac{1}{4^t-1}+\varepsilon.
\end{equation}
Since $\mathcal{P}\neq\emptyset$, either $I^{\otimes s}\notin\mathcal{P}$ or $I^{\otimes s}\in\mathcal{P}$. First consider the case where $I^{\otimes s}\notin\mathcal{P}$. In this case, we have
\begin{align}
\Pr(P_1=I^{\otimes h}|P_2\in\mathcal{P})&=\frac{\Pr(P_1=I^{\otimes h}\text{ and }P_2\in\mathcal{P})}{\Pr(P_2\in\mathcal{P})}\nonumber\\
&\leq\frac{|\mathcal{P}|\left(\frac{1}{4^t-1}+\varepsilon\right)}{|\mathcal{P}|4^{h}\left(\frac{1}{4^t-1}-\varepsilon\right)}\nonumber\\
&\leq \frac{1}{4^h}(1+3(4^t-1)\varepsilon)
\end{align}
where in the last line, we use the fact that $\frac{1+x}{1-x}\leq 1+3x$ for every $0\leq x\leq\frac{1}{3}$. Secondly, for the case where $I^{\otimes s}\in\mathcal{P}$, we have
\begin{align}
&\Pr(P_1=I^{\otimes h}|P_2\in\mathcal{P})\nonumber\\
=&\frac{\Pr(P_1=I^{\otimes h}\text{ and }P_2\in\mathcal{P})}{\Pr(P_2\in\mathcal{P})}\nonumber\\\
=&\frac{\Pr(P_1=I^{\otimes h}\text{ and }P_2=I^{\otimes s})+\Pr(P_1=I^{\otimes h}\text{ and }P_2\in\mathcal{P}\setminus\{I^{\otimes s}\})}{\Pr(P_2=I^{\otimes s})+\Pr(P_2\in\mathcal{P}\setminus\{I^{\otimes s}\})}\nonumber\\
\leq&\frac{(|\mathcal{P}|-1)(\frac{1}{4^t-1}+\varepsilon)}{(4^{h}|\mathcal{P}|-1)(\frac{1}{4^t-1}-\varepsilon)}\nonumber\\
\leq&\frac{1}{4^h}(1+3(4^t-1)\varepsilon).
\end{align}
\end{proof}

\begin{lemma}
\label{lem:compiled_operator_Z_type}
Let $r\in\{1,2,3,4\}$. Let $K_1,\ldots,K_r$ be arbitrary non-identity $t$-qubit Pauli operators. For every $i\in\{1,\ldots,r\}$, draw an independent $\tilde{C}_i\sim\tilde{\mathcal{D}}_t$, define $P_i=\tilde{C}_i^\dagger K_i\tilde{C}_i$, and write $P_i=\pm E_i\otimes F_i$ such that $E_i$ and $F_i$ are $\frac{t}{4}$- and $\frac{3t}{4}$-qubit Pauli operators respectively. Define $P_Q=\bigotimes_{i=1}^r E_i$, $P_Q'=\bigotimes_{i=1}^r F_i$, and $k=\frac{rt}{4}$. Let $\emptyset\neq\mathcal{P}\subseteq\{I,X,Y,Z\}^{3k}$. Then 
\begin{enumerate}
\item for $r=4$,
\begin{equation}
\Pr_{(\tilde{C}_1,\ldots,\tilde{C}_r,\tilde{C})\sim\tilde{\mathcal{D}}_t}[\tilde{C}^\dagger P_Q\tilde{C}\text{ is $Z$-type}|P_Q'\in\mathcal{P}]\leq\left(\frac{1}{2}\right)^{t-2^{2t+2}\varepsilon},
\end{equation}

\item for $r\in\{1,2,3\}$, for every $\varepsilon\leq \frac{1}{8^t}$,
\begin{equation}
\Pr_{(\tilde{C}_1,\ldots,\tilde{C}_r,\tilde{C})\sim\tilde{\mathcal{D}}_t}[\tilde{C}^\dagger(P_Q\otimes I^{\otimes t-k})\tilde{C}\text{ is $Z$-type}|P_Q'\in\mathcal{P}]\leq\frac{1}{4^k}+\frac{2}{2^t}.
\end{equation}
\end{enumerate}
\end{lemma}

\begin{proof}
For every non-identity $t$-qubit Pauli operator $W\in\{I,X,Y,Z\}^{\otimes t}\setminus\{I^{\otimes t}\}$, it holds that
\begin{equation}
\Pr_{\tilde{C}\sim\tilde{\mathcal{D}}_t}\left(\tilde{C}^\dagger W\tilde{C}\text{ is $Z$-type}\right)=2^t\Exp_{\tilde{C}\sim\tilde{\mathcal{D}}_t}\left[\bra{0^n}\tilde{C}^\dagger\frac{W}{\sqrt{2^t}}\tilde{C}\ket{0^n}^2\right]\leq \frac{1}{2^t+1}+2^t\varepsilon.
\end{equation}
By the previous lemma and the independence among $F_1,\ldots,F_r$, 
\begin{align}
\Pr(P_Q=I^{\otimes k}|P_Q'\in\mathcal{P})&\leq \frac{1}{4^k}(1+3(4^t-1)\varepsilon)^r\nonumber\\
&\leq \frac{1}{4^k}+\frac{15(4^t-1)}{4^k}\varepsilon
\end{align}
where we use the fact that $(1+x)^r\leq (1+x)^4\leq 1+5x$ for every $0\leq x\leq 0.15$. Thus, for $r=4$, we get by convexity that
\begin{align}
&\Pr(\tilde{C}^\dagger P_Q \tilde{C}\text{ is $Z$-type}|P_Q'\in\mathcal{P})\nonumber\\
\leq&\Pr(P_Q=I^{\otimes k}|P_Q'\in\mathcal{P})+(1-\Pr(P_Q=I^{\otimes k}|P_Q'\in\mathcal{P}))\left(\frac{1}{2^t+1}+2^t\varepsilon\right)\nonumber\\
\leq&\frac{1}{4^t}+15\varepsilon+\left(1-\frac{1}{4^t}-15\varepsilon\right)\left(\frac{1}{2^t+1}+2^t\varepsilon\right)\nonumber\\
\leq&\frac{1}{2^t}+(2^t+15)\varepsilon\nonumber\\
=& 2^{-t+\log(1+2^t(2^t+15)\varepsilon)}\nonumber\\
\leq& 2^{-t+2^{2t+2}\varepsilon}
\end{align}
where in the last line, we use the fact that $\log_2(1+x)\leq 2x$ for every $x\geq 0$. For the second bound, we have
\begin{align}
&\Pr(\tilde{C}^\dagger(P_Q\otimes I^{\otimes t-k})\tilde{C}\text{ is $Z$-type}|P_Q'\in\mathcal{P})\nonumber\\
\leq &\frac{1}{4^k}+\frac{15(4^t-1)}{4^k}\varepsilon+\left(1-\frac{1}{4^k}-\frac{15(4^t-1)}{4^k}\varepsilon\right)\left(\frac{1}{2^t+1}+2^t\varepsilon\right)\nonumber\\
\leq &\frac{1}{4^k}+\frac{1}{2^t}+\left(\frac{15(4^t-1)}{4^k}+2^t\right)\varepsilon\nonumber\\
\leq &\frac{1}{4^k}+\frac{1}{2^t}+\frac{1}{2^t}
\end{align}
since $\varepsilon\leq\frac{1}{8^t}$.
\end{proof}

\begin{theorem}
\label{thm:compiled_circuit_Z_type}
Let $P$ be an $n$-qubit Pauli operator, and let $\tilde{D}$ be a compiled random Clifford circuit. Also let $a = \abs{\mathcal{G}(P)}$ and $l = \mathrm{Per}(\mathcal{G}(P))$ denote the size and perimeter of $\mathcal{G}(P)$, respectively. Assume $\tau\geq\Omega(\sqrt{\log(m)})$ and choose $\varepsilon=2^{-ct}=2^{-c\tau^2}$ for some large enough $c$ so that $(m-1)^22^{2t+2}\varepsilon\leq 1$. Then 
\begin{align}
\Pr[\tilde{D}^{\dagger}  P \tilde{D} \text{ is $Z$-type}] \leq \left(\frac{1}{2} \right) ^{a \tau^2 + \frac{1}{4}l (\tau^2 / 12 - 2)-1}.
\end{align}
\end{theorem}
\begin{proof}
It suffices to establish bounds analogous to the ones appearing in the proof of \cref{lem:coarse_grained_prob_Ztype} for compiled random Clifford circuits. We inherit all the definitions from the proof of \cref{lem:coarse_grained_prob_Ztype}. 
For every $V_i\in\mathcal{V}_\text{int}$,
\begin{align}
\Pr\left[ V_i^\dagger P'_i V_i \text{ is $Z$-type} \; | \; V_j^\dagger P'_j V_j \text{ is $Z$-type for all $j < i $}\right] \leq 2^{-\tau^2+2^{2t+2}\varepsilon} = 2^{- \abs{Q_i'}+2^{2t+2}\varepsilon}.
\end{align}
For every $V_i\in\mathcal{V}_\text{perim}$,
\begin{align}
\Pr\left[ V_i^\dagger P'_i V_i \text{ is $Z$-type} \; | \; V_j^\dagger P'_j V_j \text{ is $Z$-type for all $j < i $}\right] &\leq 4^{-\abs{Q_i'}} + 2^{-\tau^2+1}\nonumber\\
&\leq  4^{-\abs{Q_i'}} + 2^{-4 \abs{Q_i'} / 3+1}\nonumber\\
&=2^{-\abs{Q_i'}} \left(2^{-\abs{Q_i'}} + 2^{-\abs{Q_i'}/3+1} \right)\nonumber\\
&\leq 2^{-|Q'|-\tau^2/12+2}.
\end{align}
Finally,
\begin{align}
\Pr[D^\dagger PD\text{ is $Z$-type}]&\leq 2^{-\abs{Q_0'}+2^{2t+2}\varepsilon} \left(\prod_{V_i \in \mathcal{V}_\text{int}} 2^{-\abs{Q_i'}+2^{2t+2}\varepsilon}\right)\left(\prod_{V_j \in \mathcal{V}_\text{perim}} 2^{-\abs{Q_j} - \tau^2/12 + 2}\right)\nonumber\\
&\leq 2^{- \abs{\mathcal{V}_\text{perim}} (\tau^2/12 - 2) - \sum_{i=0}^{(m-1)^2} \abs{Q'_i}+(m-1)^22^{2t+2}\varepsilon}\nonumber \\
&\leq 2^{- \abs{\mathcal{V}_\text{perim}} (\tau^2/12 - 2) - \abs{Q'}+1}.
\end{align}
To finish the proof, we recall from the proof of \cref{lem:coarse_grained_prob_Ztype} that $|Q_i|=a\tau^2$ and $|\mathcal{V}_\text{perim}|\geq \frac{1}{4}l$.
\end{proof}

\section{Proof of \cref{lem:limit} \label{sec:ghzrestriction}}
\begin{proof}
Without loss of generality let us fix a local basis so that $C_h, C_i, C_j$ are identity, and the stabilizer group of $\psi$ contains stabilizers $g_1,g_2,g_3$ in \cref{eq:g13}. By taking products of these stabilizers $g_1,g_2,g_3$  we see that the stabilizer group of $\psi$ contains the following elements:
\begin{align}
g_1&= X_hX_iX_j Z(s) \nonumber\\
g_2&=Z_hZ_i I_j Z(t)\nonumber\\
g_3&=I_h Z_iZ_j Z(u)\nonumber\\
g_4&= -Y_hY_iX_j Z(s\oplus t)\nonumber\\
g_5&=-X_hY_iY_j Z(s\oplus u)\nonumber\\
g_6&=-Y_hX_iY_j Z(s\oplus t\oplus u)\nonumber\\
g_7&=Z_hI_iZ_j Z(u\oplus t).
\label{eq:gall}
\end{align}
Let us assume, to reach a contradiction, that $m$ satisfies \cref{eq:possibilistic}. Write
\[
m_s=\prod_{v\in [n]}m_v^{s_v} \qquad m_t=\prod_{v\in [n]}m_v^{t_v} \qquad m_u=\prod_{v\in [n]}m_v^{u_v}.
\]

Note that our function $m$ must only output binary strings that are consistent with stabilizers that are diagonal in the measurement basis $b$. In this way each of the stabilizers \cref{eq:gall} places a constraint on the possible outputs of $m$. 

For example, from the stabilizer $g_2$ we infer that 
\begin{equation}
m_h\cdot m_i\cdot m_t=1 \qquad \text{whenever} \quad b_h=b_i=2.
\label{eq:stabg2}
\end{equation}
Since each output bit depends on only one input bit, and  $m_h,m_i$ are independent of $b_j$,  we can infer from the above that $m_t$ is independent of $b_j$. We now explain this in more detail. Let $L_h,L_i,L_j\subseteq [n]$ be the set of output bits $m_v$ that depend on inputs $b_h,b_i,b_j$ respectively, so that
\begin{equation}
m_t=(-1)^a\prod_{v\in L_h} m_v(b_h)\prod_{v\in L_i} m_v(b_i)\prod_{v\in L_j} m_v(b_j)
\label{eq:mt}
\end{equation}
for some $a\in \{0,1\}$. Then \cref{eq:stabg2} implies
\begin{equation}
\prod_{v\in L_j} m_v(b_j)=(-1)^a m_h(2)\cdot m_i(2)\cdot \prod_{v\in L_h} m_v(2)\prod_{v\in L_i} m_v(2)\qquad b_j\in \{0,1,2\}.
\label{eq:prodc}
\end{equation}
where $b\in \{0,1\}$ is fixed. Plugging \cref{eq:prodc} into \cref{eq:mt} gives 
\[
m_t=\prod_{v\in L_h} m_v(b_h)\prod_{v\in L_i} m_v(b_i)\left[m_h(2)\cdot m_i(2)\cdot \prod_{v\in L_h} m_v(2)\prod_{v\in L_i} m_v(2)\right]
\]
which shows that $m_t$ is independent of $b_j$. 

In exactly the same way, from the stabilizers $g_3, g_7$ we infer that $m_u$ is independent of $b_h$ and that $m_u\cdot m_t$ is independent of $b_i$.

Next consider the subset of inputs $b_h,b_i,b_j\in \{0,1\}$ which correspond to measuring these qubits in the $X$ or $Y$ bases. Since $m_t$ is independent of $b_j$ and $m_U$ is independent of $b_h$, there exist affine functions $e,f,g,h:\{0,1\}\rightarrow \{0,1\}$ such that 
\begin{equation}
m_t=(-1)^{e(b_h)+f(b_i)} \quad m_u=(-1)^{g(b_i)+h(b_j)} \qquad b_h,b_i,b_j\in \{0,1\}.
\label{eq:mtmu}
\end{equation}
Since $m_t\cdot m_u$ is independent of $b_i$ we may set 
\begin{equation}
g(b_i)=f(b_i).
\label{eq:fg}
\end{equation}

From the stabilizers $g_1,g_4,g_5,g_6$ we infer
\begin{equation}
i^{b_h+b_i+b_j}m_hm_im_j m_{s}m_t^{b_h}m_u^{b_j}=1 \qquad b_h\oplus b_i\oplus b_j=0 \qquad b_h,b_i,b_j\in \{0,1\}.
\label{eq:ms1}
\end{equation}

Plugging \cref{eq:mtmu,eq:fg} into \cref{eq:ms1} gives, for $b_h,b_i,b_j\in \{0,1\}$, 
\begin{equation}
i^{b_h+b_i+b_j}m_hm_im_j m_{s}(-1)^{b_h(e(b_h)+f(b_i))}(-1)^{b_j(f(b_i)+h(b_j))}=1 \qquad b_h\oplus b_i\oplus b_j=0.
\label{eq:ms2}
\end{equation}

Now since each output bit $m_v$ depends on at most one of the variables $b_h,b_i,b_j$, from the above we infer that for $b_h,b_i,b_j\in \{0,1\}$ we have
\begin{equation}
i^{b_h+b_i+b_j}(-1)^{\alpha_0+\alpha_1b_h+\alpha_2b_i+\alpha_3b_j}(-1)^{\beta b_hb_i}(-1)^{\beta b_ib_j}=1 \qquad b_h\oplus b_i\oplus b_j=0 
\label{eq:ms}
\end{equation}
for some coefficients
\begin{equation}
\alpha_0,\alpha_1,\alpha_2,\alpha_3,\beta \in \{0,1\}
\label{eq:coefs}
\end{equation}
The claim below shows that we have reached a contradiction and we conclude that \cref{eq:possibilistic} is not satisfied by any function $m$ of the form described in the theorem statement. \cref{eq:fails} then follows directly.
\begin{claim}
\cref{eq:ms} is not satisfied for any choice of coefficients \cref{eq:coefs}.
\end{claim}
\begin{proof}
Assume, to reach a contradiction, that \cref{eq:ms} holds. We'll  plug in values of $b_h,b_i,b_j$ and infer a set of linear equations that must be satisfied by the coefficients $\alpha_0,\alpha_1,\alpha_2,\alpha_3,\beta$. By plugging in $b_h=b_i=b_j=0$ we infer $\alpha_0=0$. Next setting $(b_h,b_i,b_j)=\{(1,1,0),(1,0,1),(0,1,1)\}$ we get
\begin{align*}
\alpha_1+\alpha_2+\beta&=1 \mod 2\\
\alpha_1+\alpha_3 &=1 \mod 2\\
\alpha_2+\alpha_3+\beta&=1 \mod 2\\
\end{align*}
Summing these three equations we are led to a contradiction ($0=1 \mod 2$) and therefore there is no solution.
\end{proof}
\end{proof}

\end{document}